\documentclass[draft,12pt,a4paper,onecolumn]{IEEEtran} 
\usepackage{amsmath,amsfonts, amssymb}

\newtheorem{thm}{Theorem}
\newtheorem{cor}[thm]{Corollary}
\newtheorem{lem}{Lemma}
\newtheorem{df}{Definition}
\newtheorem{rem}{Remark}

\newcounter{step}
\newenvironment{algorithm}[2]{%
  \begin{list}{%
      \textsf{Step #2\arabic{step}}}%
    {%
      \usecounter{step}%
      \settowidth{\labelwidth}{\textsf{#1}}%
      \addtolength{\labelwidth}{0mm}%
      \setlength{\leftmargin}{\labelwidth}%
      \setlength{\rightmargin}{0pt}%
      \setlength{\labelsep}{0pt}%
      \setlength{\parsep}{0pt}%
      \setlength{\itemsep}{0pt}%
  }}{\end{list}}

\newcommand{\e}{\varepsilon}
\newcommand{\vphi}{\varphi}
\newcommand{\qed}{$\blacksquare$}

\newcommand{\lrB}[1]{\left[{#1}\right]}
\newcommand{\lrb}[1]{\left\{{#1}\right\}}
\newcommand{\lrsb}[1]{\left({#1}\right)}
\newcommand{\lrbar}[1]{\left|{#1}\right|}
\newcommand{\limn}{\lim_{n\to\infty}}
\newcommand{\liminfn}{\liminf_{n\to\infty}}

\newcommand{\A}{\mathcal{A}}
\newcommand{\B}{\mathcal{B}}
\newcommand{\C}{\mathcal{C}}
\newcommand{\cH}{\mathcal{H}}
\newcommand{\J}{\mathcal{J}}
\newcommand{\M}{\mathcal{M}}
\newcommand{\Q}{\mathcal{Q}}
\newcommand{\R}{\mathcal{R}}
\newcommand{\cS}{\mathcal{S}}
\newcommand{\T}{\mathcal{T}}
\newcommand{\U}{\mathcal{U}}
\newcommand{\V}{\mathcal{V}}
\newcommand{\X}{\mathcal{X}}
\newcommand{\Y}{\mathcal{Y}}
\newcommand{\Z}{\mathcal{Z}}
\newcommand{\G}{\mathcal{G}}

\newcommand{\sfA}{\mathsf{A}}
\newcommand{\sfB}{\mathsf{B}}
\newcommand{\sfc}{\mathsf{c}}

\newcommand{\ba}{\boldsymbol{a}}
\newcommand{\cc}{\boldsymbol{c}}
\newcommand{\mm}{\boldsymbol{m}}
\newcommand{\bp}{\boldsymbol{p}}
\newcommand{\bt}{\boldsymbol{t}}
\newcommand{\uu}{\boldsymbol{u}}
\newcommand{\vv}{\boldsymbol{v}}
\newcommand{\xx}{\boldsymbol{x}}
\newcommand{\yy}{\boldsymbol{y}}
\newcommand{\MM}{\boldsymbol{M}}
\newcommand{\UU}{\boldsymbol{U}}
\newcommand{\VV}{\boldsymbol{V}}
\newcommand{\WW}{\boldsymbol{W}}
\newcommand{\XX}{\boldsymbol{X}}
\newcommand{\YY}{\boldsymbol{Y}}
\newcommand{\ZZ}{\boldsymbol{Z}}
\newcommand{\zero}{\boldsymbol{0}}
\newcommand{\one}{\boldsymbol{1}}
\newcommand{\aalpha}{\boldsymbol{\alpha}}
\newcommand{\bbeta}{\boldsymbol{\beta}}
\newcommand{\oomega}{\boldsymbol{\omega}}

\newcommand{\bcA}{\boldsymbol{\A}}
\newcommand{\bcB}{\boldsymbol{\B}}
\newcommand{\sfcc}{\boldsymbol{\sfc}}

\newcommand{\tX}{\widetilde{X}}
\newcommand{\tmu}{\widetilde{\mu}}

\newcommand{\hX}{\widehat{X}}
\newcommand{\hcH}{\widehat{\cH}}
\newcommand{\hXX}{\widehat{\XX}}

\newcommand{\bU}{\overline{\U}}
\newcommand{\osfA}{\overline{\sfA}}

\newcommand{\Error}{\mathrm{Error}}
\newcommand{\im}{\mathrm{Im}}
\newcommand{\GFq}{\mathrm{GF}(q)}
\newcommand{\Prob}{P}

\newcommand{\bpA}{\bp_{\sfA}}
\newcommand{\bpAB}{\bp_{\sfA\sfB}}
\newcommand{\bpB}{\bp_{\sfB}}
\newcommand{\pA}{p_{\sfA}}
\newcommand{\pB}{p_{\sfB}}
\newcommand{\pAB}{p_{\sfA\sfB}}
\newcommand{\alphaA}{\alpha_{\sfA}}
\newcommand{\alphaB}{\alpha_{\sfB}}
\newcommand{\alphaAB}{\alpha_{\sfA\sfB}}
\newcommand{\betaA}{\beta_{\sfA}}
\newcommand{\betaB}{\beta_{\sfB}}
\newcommand{\betaAB}{\beta_{\sfA\sfB}}
\newcommand{\aalphaA}{\aalpha_{\sfA}}
\newcommand{\aalphaB}{\aalpha_{\sfB}}
\newcommand{\bbetaA}{\bbeta_{\sfA}}
\newcommand{\bbetaB}{\bbeta_{\sfB}}

\newcommand{\oH}{\overline{H}}
\newcommand{\uH}{\underline{H}}
\newcommand{\oI}{\overline{I}}
\newcommand{\uI}{\underline{I}}
\newcommand{\oT}{\overline{\T}}
\newcommand{\uT}{\underline{\T}}
\newcommand{\oD}{\overline{D}}
\newcommand{\otheta}{\overline{\theta}}
\newcommand{\utheta}{\underline{\theta}}

\allowdisplaybreaks

\title{
  Channel Coding and Lossy Source Coding\\
  Using a Constrained Random Number Generator
}
\author{
  Jun~Muramatsu
  \thanks{J.~Muramatsu is with
    NTT Communication Science Laboratories, NTT Corporation,
    2-4, Hikaridai, Seika-cho, Soraku-gun, Kyoto 619-0237, Japan
    (E-mail: muramatsu.jun@lab.ntt.co.jp).
  }
}

\date{\today}
\begin{document}
\maketitle

\begin{abstract}
  Stochastic encoders for channel coding and lossy source coding
  are introduced with a rate close to the fundamental limits,
  where the only restriction is that the channel input alphabet
  and the reproduction alphabet of the lossy source code are finite.
  Random numbers, which satisfy a condition specified by a function and
  its value, are used to construct stochastic encoders.
  The proof of the theorems is based on the hash property of an
  ensemble of functions, where the results are extended to general
  channels/sources and alternative formulas are introduced for channel
  capacity and the rate-distortion region.
  Since an ensemble of sparse matrices has a hash property,
  we can construct a code by using sparse matrices,
  where the sum-product algorithm can be used for encoding and decoding
  by assuming that channels/sources are memoryless.
\end{abstract}

\section{Introduction}

The aim of this paper is to introduce a channel code and a lossy source
code for general channels/sources including additive Gaussian, Markov,
and non-stationary channels/sources.
The only assumption is that the input alphabet for channel coding and
the reproduction alphabet for lossy source coding are finite.
We prove that the fundamental limits called the channel capacity and the
boundary of the rate-distortion region are achievable with the proposed
codes.
We introduce {\em stochastic} encoders for constructing the codes
and we can easily modify these encoders to make them deterministic.
Let $\X^n$ be the cartesian power of a set $\X$,
and $\xx$ denotes an element of $\X^n$.
To construct stochastic encoders,
we use a sequence of random numbers subject to a distribution
$\tmu$ on $\X^n$ defined as
\[
  \tmu(\xx)\equiv
  \begin{cases}
    \frac{\mu(\xx)}{\mu(\{\xx: A\xx=\cc\})}
    &\text{if}\ A\xx=\cc
    \\
    0
    &\text{if}\ A\xx\neq\cc
  \end{cases}
\]
for a given probability distribution $\mu$ on $\X^n$,
a function $A:\X^n\to\{A\xx: \xx\in\X^n\}$, and $\cc\in\{A\xx:\xx\in\X^n\}$.
Let us call a generator for this type of random number
a {\em constrained random number generator}.

One contribution of this paper is to extend the results of~\cite{HASH}
to general channels/sources.
In~\cite{HASH}, the direct part of the channel coding theorem and the
lossy source coding theorem for a discrete stationary memoryless
channel/source are shown based on the hash property of an ensemble of
functions, which is an extension of random bin coding~\cite{C75}, the
set of all linear functions~\cite{CSI82}, and the two-universal class
of hash functions \cite{CW}.
In this paper, alternative general formulas for the channel capacity and
rate-distortion region are introduced and the achievability of the
proposed codes is proved based on a stronger version of hash property
introduced
in~\cite{HASH-BC}\cite{ISIT2010}\cite{ISIT2011a}\cite{ISIT2011b}.
Since an ensemble of sparse matrices has a hash property, we can
construct codes by using sparse matrices.

Another contribution of this paper is that we introduce a practical
algorithm for the proposed code for a (non-stationary) memoryless
(asymmetric) channel/source.
We introduce an practical algorithm for a constrained random generator
by using a sparse matrix and a sum-product
algorithm~\cite{GDL}\cite{KFL01},
where we assume that a channel/source is (non-stationary) memoryless.
There are many ways to construct channel
codes~\cite{BB04}\cite{GA62}\cite{Mac99}
and lossy source codes~\cite{GV09}\cite{M04}\cite{MY03}\cite{WM09}
by using sparse matrices.
These approaches assume that a channel/source is stationary memoryless
and symmetric, or a quantization map~\cite[Section 6.2]{GA68}
is used for an asymmetric channel/source.
On the other hand, the only requirement for the proposed code is that
the input alphabet for channel coding and the reproduction alphabet for
lossy source coding are finite.

It should be noted that a similar idea has appeared
in~\cite{RR11}\cite{YAG12},
where they introduced random bin coding (privacy amplification)
and Slepian-Wolf decoding\footnote{It should be noted that the idea of
  using Slepian-Wolf decoding has already been mentioned
  in~\cite{SWLDPC}\cite{HASH}.}
(information reconciliation) for the construction of codes,
and their proofs are based on the fact that the output statistics of
random binning are uniformly distributed.
Furthermore, the encoding functions seem to be impractical.
This paper describes the explicit practical construction of encoding functions
and theorems are proved simply and rigorously based on the technique
reported in \cite{ISIT2011a}, where it is proved that we can use sparse
matrices for the construction of codes.

This paper is organized as follows.
Section~\ref{sec:review} reviews formulas for the channel capacity and
the rate-distortion region based on the information spectrum method
introduced in~\cite{HV93}\cite{HAN}\cite{VH94}.
Alternative formulas for the channel capacity and the rate-distortion
region are also introduced.
Section~\ref{sec:hash} describes the notion of a hash property,
which is stronger than that introduced in~\cite{HASH}.
Several lemmas are introduced that will be used in the proof of the
theorems.
Section~\ref{sec:channel} deals with the construction of a channel code
and Section~\ref{sec:lossy} describes the construction of a lossy code.
Section~\ref{sec:sum-product} describes an algorithm for a constrained
random number generator by using a sum-product algorithm.
The conversion from stochastic encoders into deterministic encoders
is discussed in this section.
Theorems and lemmas are proved in Section~\ref{sec:proof}.
Some lemmas are shown in Appendix.

\section{Formal Description of Problems and
  General Formulas for Channel Capacity and Rate Distortion Region}
\label{sec:review}

This section provides a formal description of the problems and reviews
formulas for the channel capacity and the rate distortion region.
All the results in this paper are presented by using the
information spectrum method introduced
in~\cite{HV93}\cite{HAN}\cite{VH94},
where the consistency and stationarity of channels/sources are not
assumed.
It should be noted that all the results reported in this paper can be
applied to stationary ergodic channels/sources and stationary memoryless
channels/sources.

Throughout this paper, we denote the probability of an event by
$\Prob(\cdot)$ and denote the probability distribution of a random
variable $U$ by $\mu_U$.

We call a sequence $\UU\equiv\{U^n\}_{n=1}^{\infty}$ of random variables
a {\em general source}, where $U^n\in\U^n$.
For a general source $\UU$, we define the spectral sup-entropy rate
$\oH(\UU)$ and the spectral inf-entropy rate $\uH(\UU)$ as
\begin{align*}
  \oH(\UU)
  &\equiv\inf\lrb{
    \theta: \limn\Prob\lrsb{ \frac 1n\log\frac1{\mu_{U^n}(U^n)} > \theta  }=0
  }
  \\
  \uH(\UU)
  &\equiv\sup\lrb{
    \theta: \limn\Prob\lrsb{ \frac 1n\log\frac1{\mu_{U^n}(U^n)} < \theta  }=0
  }.
\end{align*}
It is known that both $\oH(\UU)$ and $\uH(\UU)$ are equal to the entropy
rate of $\UU$ when $\UU$ is stationary ergodic,
that is,
\[
  \oH(\UU)=\uH(\UU)=\limn \frac{H(U^n)}n,
\]
where $H(U^n)$ is the entropy of $U^n$.
When $\UU$ is stationary memoryless,
both $\oH(\UU)$ and $\uH(\UU)$ are equal to the entropy $H(U^1)$.

For a pair $(\UU,\VV)=\{(U^n,V^n)\}_{n=1}^{\infty}$ of general sources,
we define the spectral conditional sup-entropy rate $\oH(\UU|\VV)$,
the spectral conditional inf-entropy rate $\uH(\UU|\VV)$,
the spectral sup-mutual information rate $\oI(\UU;\VV)$,
and the spectral inf-mutual information rate $\uI(\UU;\VV)$ as
\begin{align*}
  \oH(\UU|\VV)
  &\equiv\inf\lrb{
    \theta: \limn\Prob\lrsb{ \frac 1n\log\frac1{\mu_{U^n|V^n}(U^n|V^n)} > \theta  }=0
  }
  \\
  \uH(\UU|\VV)
  &\equiv\sup\lrb{
    \theta: \limn\Prob\lrsb{ \frac 1n\log\frac1{\mu_{U^n|V^n}(U^n|V^n)} < \theta  }=0
  }
  \\
  \oI(\UU;\VV)
  &\equiv\inf\lrb{
    \theta: \limn\Prob\lrsb{
      \frac 1n\log
      \frac{\mu_{U^nV^n}(U^n,V^n)}{\mu_{U^n}(U^n)\mu_{V^n}(V^n)}
      > \theta
    }=0
  }
  \\
  \uI(\UU;\VV)
  &\equiv\sup\lrb{
    \theta: \limn\Prob\lrsb{
      \frac 1n\log
      \frac{\mu_{U^nV^n}(U^n,V^n)}{\mu_{U^n}(U^n)\mu_{V^n}(V^n)}
      < \theta
    }=0
  },
\end{align*}
where $\mu_{U^nV^n}$ is the joint probability distribution corresponding
to $(U^n,V^n)$.
It is known that both $\oH(\UU|\VV)$ and $\uH(\UU|\VV)$ are equal to the
conditional entropy rate of $\UU$ given $\VV$,
and both  $\oI(\UU;\VV)$ and $\uI(\UU;\VV)$ are equal to the mutual
information rate between $\UU$ and $\VV$,
when $(\UU,\VV)$ is stationary ergodic, that is,
\begin{gather*}
  \oH(\UU|\VV)=\uH(\UU|\VV)=\limn\frac{H(U^n|V^n)}n
  \\
  \oI(\UU;\VV)=\uI(\UU;\VV)=\limn\frac{I(U^n;V^n)}n,
\end{gather*}
where $H(U^n|V^n)$ is the conditional entropy of $U^n$ given $V^n$
and $I(U^n;V^n)$ is the mutual information between $U^n$ and $V^n$.
When $(\UU,\VV)$ is stationary memoryless,
both $\oH(\UU|\VV)$ and $\uH(\UU|\VV)$ are equal to the conditional
entropy $H(U^1|V^1)$ and both  $\oI(\UU;\VV)$ and $\uI(\UU;\VV)$
are equal to the mutual information $I(U^1;V^1)$.

\subsection{Channel Capacity}

In the following, we introduce the definition of the channel capacity
for a general channel.
Let $\X^n$ and $\Y^n$ be the alphabets of a channel input $X^n$ and a
channel output $Y^n$, respectively.
A sequence $\WW\equiv\{p_{Y^n|X^n}\}_{n=1}^{\infty}$
of conditional probability distributions is called a
{\em general channel}.
\begin{df}
For a general channel $\WW$,
we call a rate $R$ {\em achievable} if for all $\delta>$ and all
sufficiently large $n$ there is a pair consisting of an encoder
$\vphi_n:\M_n\to\X^n$ and a decoder $\psi_n:\Y^n\to\M_n$ such that
\begin{gather*}
  \frac 1n\log|\M_n|\geq R
  \\
  P(\psi_n(Y^n)\neq M_n)\leq\delta,
\end{gather*}
where $[1/n]\log|\M_n|$ represents the rate of the code,
$M_n$ is a random variable of the message corresponding to the uniform
distribution on $\M_n$ and the joint distribution $\mu_{M_nY^n}$ is
given as
\[
  \mu_{M_nY^n}(\mm,\yy)\equiv\frac{\mu_{Y^n|X^n}(\yy|\vphi_n(\mm))}{|\M_n|}.
\]
The {\em channel capacity} $C(\WW)$ is defined by the supremum of the
achievable rate.
\end{df}

For a general channel $\WW$, the channel capacity $C(\WW)$ is derived
in~\cite{VH94} as
\begin{equation}
  C(\WW)=\sup_{\XX}\uI(\XX;\YY),
  \label{eq:capacity-I}
\end{equation}
where the supremum is taken over all general sources
$\XX=\{X^n\}_{n=1}^{\infty}$ and the joint distribution $\mu_{X^nY^n}$
is given as
\begin{equation}
  \mu_{X^nY^n}(\xx,\yy)
  \equiv
  \mu_{Y^n|X^n}(\yy|\xx)\mu_{X^n}(\xx).
  \label{eq:channel-joint}
\end{equation}

We introduce the following lemma, which will be proved in
Section~\ref{sec:proof-capacity}.
It should be noted that this capacity formula is a straightforward
generalization of that obtained by Shannon~\cite{SHANNON}.
\begin{lem}
\label{lem:capacity}
For a general channel $\WW$,
\begin{equation}
  C(\WW)=\sup_{\XX}\lrB{\uH(\XX)-\oH(\XX|\YY)},
  \label{eq:capacity}
\end{equation}
where the supremum is taken over all general sources $\XX$ and the joint
distribution of $(\XX,\YY)$ is given as (\ref{eq:channel-joint}).
\end{lem}

\begin{rem}
When $\WW$ is stationary ergodic, it is sufficient that
the supremum on the right hand side of
(\ref{eq:capacity-I}) and  (\ref{eq:capacity})
is taken over all stationary ergodic sources.
When $\WW$ is stationary memoryless, it is sufficient that
the supremum on the right hand side of
(\ref{eq:capacity-I}) and  (\ref{eq:capacity})
is taken over all stationary memoryless sources.
For these reasons, the lemma is trivial in these cases.
\end{rem}

In this paper, we construct a channel code whose rate is close to the
channel capacity given by (\ref{eq:capacity}).
Constructed code is given by a pair consisting of a stochastic encoder
$\Phi_n:\M_n\to\X^n$ and a decoder $\psi_n:\Y^n\to\M_n$,
where $\M_n$ is a set of messages.

\begin{rem}
\label{rem:capasity-stochastic}
It should be noted that the capacity formulas (\ref{eq:capacity-I})
and (\ref{eq:capacity}) are satisfied when a stochastic encoder is
allowed.
In fact, by considering the average over stochastic encoders
and using the random coding argument,
we can construct a deterministic encoder from a stochastic encoder.
Thus the rate of the stochastic encoder should be upper bounded
by the channel capacity.
On the other hand,
the channel capacity is achievable with a stochastic encoder
because a deterministic encoder is one type of stochastic encoder.
\end{rem}

\subsection{Rate-Distortion Region}

In the following, we introduce the achievable rate-distortion region
for a general source.
Let $\Y^n$ be a source alphabet and $\X^n$ be a reproduction
alphabet\footnote{It should be noted that the roles of $\X^n$ and $\Y^n$
  are the reverse of those in the conventional definition of the
  rate-distortion theory.}.
Let $d_n:\X^n\times\Y^n\to[0,\infty)$ be a distortion function.

\begin{df}[{\cite[Def.\ 5.3.1]{HAN}}]
We call a pair $(R,D)$ consisting of a rate $R$ and a distortion $D$
{\em achievable} if for all $\delta>0$ and all sufficiently large $n$
there is a pair consisting of an encoder $\vphi_n:\Y^n\to\M_n$ and
a decoder $\psi_n:\M_n\to\Y^n$ such that
\begin{gather}
  \frac1n\log|\M_n|
  \leq R
  \label{eq:rd-R}
  \\
  P\lrsb{d_n(\psi_n(\vphi_n(Y^n)),Y^n)>D}
  \leq \delta.
  \label{eq:rd-D}
\end{gather}
The {\em achievable rate-distortion region} $\R(\YY)$ is defined by the
set of all achievable pairs $(R,D)$.
\end{df}
\begin{rem}
It should be noted that the factor $1/n$ appears
in~\cite[Def.~5.3.1]{HAN}.
This difference is not essential because we can replace $d_n$ by
$[1/n]d_n$ throughout this paper.
\end{rem}

For a pair $(\XX,\YY)$ of general sources, let $\oD(\XX,\YY)$ be defined as
\[
  \oD(\XX,\YY)
  \equiv
  \inf\lrb{
    \theta: 
    \limn P\lrsb{d_n(X^n,Y^n)>\theta}=0
  }.
\]

For a general source $\YY$, the rate-distortion region $\R(\YY)$ is
derived in \cite{SV96}\cite[Theorem 5.4.1]{HAN}\footnote{The
  rate-distortion function, which is the infimum of $R$ such that $(R,D)$
  is achievable for a given $D$,
  is derived in \cite{SV96}\cite[Theorem 5.4.1]{HAN}.}
as
\begin{equation}
  \R(\YY)
  =
  \bigcup_{\WW}\lrb{
    (R,D):
    \begin{gathered}
      \oI(\XX;\YY)\leq R
      \\
      \oD(\XX;\YY)\leq D
    \end{gathered}
  },
  \label{eq:rd-I}
\end{equation}
where the union is taken over all general channels 
$\WW\equiv\{\mu_{X^n|Y^n}\}_{n=1}^{\infty}$ and the joint distribution
$\mu_{X^nY^n}$ is given as
\begin{equation}
  \mu_{X^nY^n}(\xx,\yy)
  \equiv
  \mu_{X^n|Y^n}(\xx|\yy)\mu_{Y^n}(\yy).
  \label{eq:lossy-joint}
\end{equation}
We introduce the following lemma, which is proved in
Section~\ref{sec:proof-rate-distortion}.
\begin{lem}
\label{lem:rd}
For a general source $\YY$,
\begin{equation}
  \R(\YY)
  =
  \bigcup_{\WW}\lrb{
    (R,D):
    \begin{gathered}
      \oH(\XX)-\uH(\XX|\YY)\leq R
      \\
      \oD(\XX;\YY)\leq D
    \end{gathered}
  },
  \label{eq:rd}
\end{equation}
where the union is taken over all channels $\WW$ and the joint
distribution of $(\XX,\YY)$ is given as (\ref{eq:lossy-joint}).
\end{lem}

\begin{rem}
When $\XX$ is stationary ergodic, it is sufficient that
the union on the right hand side of (\ref{eq:rd-I}) and (\ref{eq:rd})
is taken over all stationary ergodic channels.
When $\XX$ is stationary memoryless, it is sufficient that
the union on the right hand side of (\ref{eq:rd-I}) and (\ref{eq:rd})
is taken over all stationary memoryless channels.
For these reasons, the lemma is trivial in these cases.
\end{rem}

In this paper, we construct a fixed-rate lossy source code,
where $(R,D)$ is close to the boundary of the region given by the right
hand side of (\ref{eq:rd}).
Constructed code is given by a pair consisting of a stochastic encoder
$\Phi_n:\X_n\to\M_n$ and a decoder $\psi_n:\M_n\to\Y^n$,
where $\M_n$ is a set of codewords.

\begin{rem}
Similarly to Remark~\ref{rem:capasity-stochastic},
formulas (\ref{eq:rd-I}) and (\ref{eq:rd}) are satisfied when a
stochastic encoder is allowed.
In fact, by considering the average over the stochastic encoders
and using the random coding argument,
we can construct a deterministic encoder from a stochastic encoder
without any loss of encoding rate. Thus the rate-distortion pair of
the stochastic encoder should be in the rate-distortion region.
On the other hand, the rate-distortion region is achievable with a
stochastic encoder because a deterministic encoder is one type of
stochastic encoder.
\end{rem}

\begin{rem}
It should also be noted that we have similar results that are obtained
in this paper by assuming 
\[
  d_{\max}\equiv\max_{n,\xx,\yy}d_n(\yy,\xx)<\infty,
\]
where (\ref{eq:rd-D}) is replaced by the average distortion criterion
\[
  E_{Y^n}\lrB{d_n(\psi_n(\vphi_n(Y^n)),Y^n)}\leq D+\delta.
\]
\end{rem}

\section{$(\aalpha,\bbeta)$-hash property}
\label{sec:hash}

In this section, we introduce the hash
property\footnote{In~\cite{HASH-BC}\cite{ISIT2010}\cite{ISIT2011a}\cite{ISIT2011b},
  it is called the `strong hash property.'
  Throughout this paper, we call it simply the `hash property.'}
introduced in~\cite{HASH-BC}\cite{ISIT2010}\cite{ISIT2011a}\cite{ISIT2011b}
and its implications.

Throughout this paper, we use the following definitions and notations.
The set $\U\setminus\V\equiv\U\cap\V^c$ denotes the set difference.
Let $A\uu$ denote a value taken by a function $A:\U^n\to\bU$ at $\uu\in\U^n$,
where $\U^n$ is the domain of $A$ and $\bU$ is the region of $A$.
It should be noted that $A$ may be nonlinear.
When $A$ is a linear function expressed by an $l\times n$ matrix,
we assume that $\U\equiv\GFq$ is a finite field and the range of
functions is $\U^l$.
For a set $\A$ of functions, let $\im A$ and $\im \A$ be defined as
\begin{align*}
  \im A&\equiv\{A\uu: \uu\in\U^n\}
  \\
  \im\A &\equiv \bigcup_{A\in\A}\im A.
\end{align*}
We define a set $\C_A(\cc)$ and $\C_{AB}(\cc,\mm)$ as
\begin{align*}
  \C_A(\cc) &\equiv\{\uu: A\uu = \cc\}
  \\
  \C_{AB}(\cc,\mm) &\equiv\{\uu: A\uu = \cc, A\uu=\mm\}.
\end{align*}
In the context of linear codes,
$\C_A(\cc)$ is called a coset determined by $\cc$.
The random variables of a function $A$ and a vector $\cc\in\im A$ are denoted
by the sans serif letters $\sfA$ and $\sfcc$, respectively.
It should be noted that the random variable of a $n$-dimensional vector
$\uu\in\U^n$ is denoted by the Roman letter $U^n$ that does not
represent a function, which is the way it has been used so far.

\subsection{Formal Definition and Basic Properties}
Here, we introduce the hash property for an ensemble of functions.
It requires stronger conditions than those introduced in \cite{HASH}.
\begin{df}
Let $\bcA\equiv\{\A_n\}_{n=1}^{\infty}$ be a sequence of sets such
that $\A_{n}$ is a set of functions $A:\U^n\to\im\A_n$.
For a probability distribution $p_{\sfA,n}$ on $\A_n$,
we call a sequence $(\bcA,\bpA)\equiv\{(\A_n,p_{\sfA,n})\}_{n=1}^{\infty}$
an {\em ensemble}.
Then, $(\bcA,\bpA)$ has an $(\aalphaA,\bbetaA)$-{\em hash property}
if there are two sequences
$\aalphaA\equiv\{\alphaA(n)\}_{n=1}^{\infty}$ and
$\bbetaA\equiv\{\betaA(n)\}_{n=1}^{\infty}$,
depending on $\{p_{\sfA,n}\}_{n=1}^{\infty}$,
such that
\begin{align}
  &\limn \alphaA{}(n)=1
  \tag{H1}
  \label{eq:alpha}
  \\
  &\limn \betaA{}(n)=0
  \tag{H2}
  \label{eq:beta}
\end{align}
and
\begin{align}
  \sum_{\substack{
      \uu'\in\U^n\setminus\{\uu\}:
      \,p_{\sfA,n}(\{A: A\uu = A\uu'\})>\frac{\alphaA{}(n)}{|\im\A_n|}
  }}
  p_{\sfA,n}\lrsb{\lrb{A: A\uu = A\uu'}}
  \leq
  \betaA{}(n)
  \tag{H3}
  \label{eq:hash}
\end{align}
for any $n$ and $\uu\in\U^n$.
Throughout this paper, we omit the dependence of $\A$, $\pA$,
$\alphaA{}$ and $\betaA{}$ on $n$.
\end{df}
\begin{rem}
In~\cite{HASH}\cite{ISIT2010}, an ensemble is required to satisfy the
condition
\begin{equation*}
  \limn \frac1n\log\frac{|\bU_n|}{|\im\A_n|}=0,
\end{equation*}
where $\bU_n$ is the range of functions.
This condition is omitted because it is unnecessary for the results
reported in this paper.
\end{rem}

Let us remark on the condition (\ref{eq:hash}).
This condition requires the sum of the collision probabilities
$p_{\sfA}\lrsb{\lrb{A: A\uu = A\uu'}}$, which is greater than
$\alpha_A/|\im\A|$, to be bounded by $\beta_A$, where the sum is taken
over all $\uu'$ except $\uu$. An intuitive interpretation of
(\ref{eq:hash}) will be provided in Section~\ref{sec:linear} by using an
ensemble of sparse matrices.
It should be noted that this condition implies
\begin{equation}
  \sum_{\substack{
      \uu\in\T
      \\
      \uu'\in\T'
  }}
  \pA\lrsb{\lrb{A: A\uu = A\uu'}}
  \leq
  |\T\cap\T'|
  +
  \frac{|\T||\T'|\alphaA{}}{|\im\A|}
  +
  \min\{|\T|,|\T'|\}\betaA{}
  \tag{H3'}
  \label{eq:whash}
\end{equation}
for any $\T,\T'\subset\U^n$, which is introduced in \cite{HASH}.
A stronger condition (\ref{eq:hash}) is required for
Lemmas~\ref{lem:hash-AB} and~\ref{lem:BCP}, which appear later.
The proof of (\ref{eq:whash}) is given in Appendix~\ref{sec:proof-whash}
for the completeness of this paper.

It should be noted that when $\A$ is a two-universal class of hash
functions \cite{CW} and $\pA$ is the uniform distribution on $\A$,
then $(\bcA,\bpA{})$ has a $(\one,\zero)$-hash property,
where $\one$ and $\zero$ denote the constant sequences of $1$ and $0$,
respectively.
Random bin coding \cite{C75} and the set of all linear
functions~\cite{CSI82} are examples of the two-universal class of hash
functions.
An ensemble of sparse matrices satisfying a hash property
is introduced in Section~\ref{sec:linear}.

We have the following lemma, where it is unnecessary to assume the
linearity of functions assumed in \cite{HASH}.
The proof is given in Appendix~\ref{sec:proof-hash-AB} for the
completeness of this paper.

\begin{lem}[{\cite[Lemma 4]{HASH-BC}}]
\label{lem:hash-AB}
Let $(\bcA,\bpA)$ and $(\bcB,\bpB)$
be ensembles satisfying an $(\aalphaA,\bbetaA)$-hash property
and an $(\aalphaB,\bbetaB)$-hash property, respectively.
Let $\A\in\bcA$ (resp. $\B\in\bcB$) be a set of functions
$A:\U^n\to\im\A$ (resp. $B:\U^n\to\im\B$).
Let $(A,B)\in\A\times\B$ be a function defined as
\[
  (A,B)\uu\equiv(A\uu,B\uu)\quad\text{for each}\ \uu\in\U^n.
\]
Let $\pAB$  be a joint distribution on $\A\times\B$ defined as
\[
  \pAB(A,B)\equiv \pA(A)\pB(B)\quad\text{for each}\ (A,B)\in\A\times\B.
\]
Then the ensemble $(\bcA\times\bcB,\bpAB)$ has an
$(\aalpha_{\sfA\sfB},\bbeta_{\sfA\sfB})$-hash property,
where $(\alphaAB,\betaAB)$ is defined as
\begin{align*}
  \alphaAB&\equiv \alphaA\alphaB
  \\
  \betaAB&\equiv \betaA+\betaB.
\end{align*}
\end{lem}

The following lemma is related to the
{\it collision-resistance property}, that is,
if the number of bins is greater than the number of items
then there is an assignment such that every bin contains at most one item.
The proof is given in Appendix~\ref{sec:proof-CRP} for the completeness
of this paper.
\begin{lem}[{\cite[Lemma 1]{HASH}}]
\label{lem:CRP}
If $(\A,\pA)$ satisfies (\ref{eq:whash}), then
\[
  \pA\lrsb{\lrb{
      A: \lrB{\G\setminus\{\uu\}}\cap\C_A(A\uu)\neq \emptyset
  }}
  \leq 
  \frac{|\G|\alphaA{}}{|\im\A|} + \betaA{}
\]
for all $\G\subset\U^n$  and $\uu\in\U^n$.
\end{lem}
We show the collision-resistance property from Lemma~\ref{lem:CRP}.
Let $\mu_U$ be the probability distribution on $\G\subset\U^n$.
We have
\begin{align}
  E_{\sfA}\lrB{
    \mu_U\lrsb{\lrb{\uu:
	\lrB{\G\setminus\{\uu\}}\cap\C_{\sfA}(\sfA\uu)\neq\emptyset}}
  }
  &
  \leq
  \sum_{\uu\in\G}\mu_U(\uu)
  p_{\sfA}\lrsb{\lrb{A: \lrB{\G\setminus\{\uu\}}\cap\C_A(A\uu)\neq\emptyset}}
  \notag
  \\
  &
  \leq
  \sum_{\uu\in\G}\mu_U(\uu)
  \lrB{\frac{|\G|\alphaA{}}{|\im\A|} + \betaA{}}
  \notag
  \\
  &
  \leq
  \frac{|\G|\alphaA{}}{|\im\A|} + \beta_A{}.
\end{align}
By assuming that $|\G|/|\im\A|$ vanishes as $n\to\infty$,
we have the fact that there is a function $A$ such that
\begin{equation*}
  \mu_U\lrsb{\lrb{\uu:
      \lrB{\G\setminus\{\uu\}}\cap\C_A(A\uu)\neq\emptyset}}
  <\delta
\end{equation*}
for any $\delta>0$ and sufficiently large $n$.
Since the relation $\lrB{\G\setminus\{\uu\}}\cap\C_A(A\uu)\neq\emptyset$
corresponds to an event where there is $\uu'\in\G$ such that $\uu$ and
$\uu'$ are different members of the same bin
(they have the same codeword determined by $A$),
we have the fact that the members of $\G$ are located in different bins
(the members of $\G$ can be decoded correctly)
with probability close to one.

The following lemma is related to the {\em balanced coloring property},
which is analogous to \cite[Lemma 3.1]{AC98}\cite[Lemma 17.3]{CK11}.
This lemma implies that there is a function $A$ such that $\T$ is almost
equally partitioned by $A$ with respect to a measure $Q$.
We use this property instead of the
{\it saturation property}~\cite{HASH},
that is, if the number of bins is greater than the number of items
there is an assignment such that every bin contains at least one item.
The proof is given in Appendix~\ref{sec:proof-BCP} for the completeness
of this paper.
\begin{lem}[{\cite[Lemma 4]{ISIT2011a}}]
\label{lem:BCP}
If $(\A,p_{\sfA})$ satisfies (\ref{eq:hash}), then
\begin{equation}
  E_{\sfA}\lrB{
    \sum_{\cc}
    \left|
      \frac{Q\lrsb{\T\cap\C_{\sfA}(\cc)}}{Q(\T)}
      -\frac 1{|\im\A|}
    \right|
  }
  \leq
  \sqrt{
    \alphaA{}-1
    +\frac {[\betaA{}+1]|\im\A|\max_{\uu\in\T} Q(\uu)}{Q(\T)}
  }
  \label{eq:BCP}
\end{equation}
for any function $Q:\U^n\to[0,\infty)$ and $\T\subset\U^n$, where
\[
  Q(\T)\equiv\sum_{\uu\in\T}Q(\uu).
\]
\end{lem}
\begin{rem}
In \cite[Lemma 3.1]{AC98}\footnote{See also \cite[Remark on Lemma B.1]{CN03}.}
and \cite[Lemma 17.3]{CK11},
the absolute value on the left hand side of (\ref{eq:BCP}) is
upper-bounded by $\e/|\im\A|$ for all $\cc\in\im\A$ and $Q\in\Q$
provided that $\e^2>3|\im\A|\log(2|\im\A||\Q|)\max_{\uu\in\T}Q(\uu)$,
where $\Q$ is a finite set of probability distributions.
\end{rem}

We show the balanced coloring property.
From Lemma~\ref{lem:BCP}, we have the fact that there is a function $A$
such that
\[
  \sum_{\cc}
  \left|
    \frac{Q\lrsb{\T\cap\C_{A}(\cc)}}{Q(\T)}
    -\frac 1{|\im\A|}
  \right|
  \leq
  \sqrt{
    \alphaA{}-1
    +\frac {[\betaA{}+1]|\im\A|\max_{\uu\in\T} Q(\uu)}{Q(\T)}
  }.
\]
By assuming that $Q(\T)\leq 1$ and $|\im\A|\max_{\uu\in\T} Q(\uu)$
vanishes as $n\to\infty$, we have
\begin{align}
  \left|
    Q\lrsb{\T\cap\C_{A}(\cc)}
    -\frac {Q(\T)}{|\im\A|}
  \right|
  &\leq
  \sum_{\cc}
  \left|
    Q\lrsb{\T\cap\C_{A}(\cc)}
    -\frac {Q(\T)}{|\im\A|}
  \right|
  \notag
  \\
  &=
  Q(\T)
  \sum_{\cc}
  \left|
    \frac{Q\lrsb{\T\cap\C_{\sfA}(\cc)}}{Q(\T)}
    -\frac 1{|\im\A|}
  \right|
  \notag
  \\
  &\leq
  \sqrt{
    \alphaA{}-1
    +[\betaA{}+1]|\im\A|\max_{\uu\in\T} Q(\uu)
  }
  \notag
  \\
  &\leq
  \delta
\end{align}
for all $\cc\in\im\A$, $\delta>0$, and sufficiently large $n$.
Since $\{\T\cap\C_A(\cc)\}_{\cc\in\im\A}$ is a partition of $\T$,
we have the fact that the set $\T$ is almost equally partitioned
with respect to a measure $Q$,
where $\cc$ represents the color of a set $\T\cap\C_A(\cc)$.

\subsection{Hash Property for Ensembles of Matrices}
\label{sec:linear}
In the following, we discuss the hash property for an ensemble of
matrices.

In the last section we discussed that the uniform distribution on
the set of all linear functions has a strong $(\one,\zero)$-hash
property because it is a universal class of hash functions.
In the following, we introduce another ensemble of matrices.

First, we introduce the average spectrum of an ensemble of
matrices given in~\cite{BB04}.
Let $\U$ be a finite field and $\A$ be a set of linear functions
$A:\U^n\to\U^l$.
It should be noted again that $A$ can be expressed
by an $l\times n$ matrix.

Let $\bt(\uu)$ be the type\footnote{In \cite{HASH}, it is called a
  histogram that is characterized by the number of occurrences of each
  symbol in the sequence $\uu$.
  The type and the histogram are essentially the same when $n$ is fixed.}
of $\uu\in\U^n$,
which is characterized by the empirical probability distribution
of the sequence $\uu$.
Let $\cH$ be a set of all types of length $n$ except $\bt(\zero)$,
where $\zero$ is the zero vector.
For a probability distribution $\pA$ on a set of $l\times n$ matrices
and a type $\bt$, let $S(\pA,\bt)$ be defined as
\[
  S(\pA,\bt)
  \equiv
  \sum_{A\in\A}\pA(A)|\{\uu\in\U^n: A\uu=\zero, \bt(\uu)=\bt\}|,
\]
which is called the expected number of codewords that have type $\bt$
in the context of linear codes.
For a given $\hcH_{\sfA}\subset\cH$, we define $\alphaA(n)$ and
$\betaA(n)$ as
\begin{align}
  \alphaA(n)
  &\equiv
  \frac{|\im\A|}{|\U|^l}\cdot\max_{\bt\in \hcH_{\sfA}}
  \frac {S(\pA,\bt)}{S(p_{\osfA},\bt)}
  \label{eq:alpha-linear}
  \\
  \betaA(n)
  &\equiv
  \sum_{\bt\in \cH\setminus\hcH_{\sfA}}S(\pA,\bt),
  \label{eq:beta-linear}
\end{align}
where $p_{\osfA}$ denotes the uniform distribution on the set of all
$l\times n$ matrices.

The following lemma provides a sufficient condition for an ensemble of
matrices to satisfy a strong hash property.
The proof is given in Appendix~\ref{sec:proof-linear}
for the completeness of this paper.
\begin{lem}[{\cite[Theorem 1]{HASH-BC}}]
\label{thm:hash-linear}
Let $(\bcA,\bpA)$ be an ensemble of matrices and assume that
$\pA\lrsb{\lrb{A: A\uu=\zero}}$ depends on $\uu$ only through the type
$\bt(\uu)$.
If $(\aalphaA,\bbetaA)$, defined by (\ref{eq:alpha-linear}) and
(\ref{eq:beta-linear}), satisfies (\ref{eq:alpha}) and (\ref{eq:beta}),
then $(\bcA,\bpA)$ has a strong $(\aalphaA,\bbetaA)$-hash property.
\end{lem}

Next, we introduce the ensemble of $q$-ary sparse matrices introduced
in~\cite{HASH}, which is the $q$-ary extension of the ensemble proposed
in~\cite{Mac99}.
Let $\U\equiv\GFq$ and $l\equiv nR$ when $0<R<1$ is given,
where $q$ is a prime number or a power of a prime number.
We generate an $l\times n$ matrix $A$ with the following procedure,
where at most $\tau$ random nonzero elements are introduced in every row.
\begin{enumerate}
  \item Start from an all-zero matrix.
  \item For each $i\in\{1,\ldots,n\}$, repeat the following procedure
  $\tau$ times:
  \begin{enumerate}
    \item Choose
    $(j,a)\in\{1,\ldots,l\}\times[\GFq\setminus\{0\}]$
    uniformly at random.
    \item Add\footnote{It should be noted that
      $(j,i)$-element of the matrix is not overwritten by $a$
      when the same $j$ is chosen again.}
    $a$ to the  $(j,i)$-element of $A$.
  \end{enumerate}
\end{enumerate}
Assume that $\tau=O(\log n)$ is even and let $(\bcA,\bpA)$ be an
ensemble corresponding to the above procedure.
Let $\hcH_{\sfA}\subset\cH$ be a set of types satisfying the requirement 
that the weight (the number of occurrences of non-zero elements)
is large enough.
Let $(\aalphaA,\bbetaA)$ be defined by (\ref{eq:alpha-linear})
and (\ref{eq:beta-linear}).
Then $\alphaA$ measures the difference between the ensemble
$(\bcA,\bpA)$ and the ensemble of all $l\times n$ matrices with respect
to the high-weight part of the average spectrum,
and $\betaA$ provides the upper bound of the probability that the code
$\{\uu\in\U^n: A\uu=\zero\}$ has low-weight codewords.
It is proved in \cite[Theorem 2]{HASH} that $(\aalphaA,\bbetaA)$ satisfy
(\ref{eq:alpha}) and (\ref{eq:beta}) if we adopt an appropriate
$\hcH_{\sfA}$.
Then, from Lemma~\ref{thm:hash-linear}, we have the fact that this
ensemble has a strong $(\aalphaA,\bbetaA)$-hash property.
It should be noted that the convergence speed of $(\aalphaA,\bbetaA)$
depends on how fast $\tau$ grows in relation to the block length.
The analysis of $(\aalphaA,\bbetaA)$ is given in the proof
of~\cite[Theorem 2]{HASH}.

\section{Construction of Channel Code}
\label{sec:channel}

This section introduces a channel code.
The idea for the construction is drawn
from~\cite{SWLDPC}\cite{HASH}\cite{ISIT2011a}.
It should be noted that we assume that the channel input alphabet $\X^n$
is a finite set but allow the channel output alphabet $\Y^n$ to be an
arbitrary (infinite, continuous) set.

For given $r>$ and $R>0$,
let $(\bcA,\bpA)$ and $(\bcB,\bpB)$ be ensembles of functions
$A:\X^n\to\im\A$ and $B:\X^n\to\im\B$ satisfying
\begin{align*}
  r&=\frac1n\log|\im\A|
  \\
  R&=\frac1n\log|\im\B|,
\end{align*}
respectively, where we define $\M_n\equiv\im\B$ and $R$ represents the
rate of the code.
We fix functions $A\in\A$, $B\in\B$, and a vector $\cc\in\im\A$
so that they are available for constructing an encoder and a decoder.

We use a constraint random number generator to construct an encoder.
Let $\tX^n\equiv\tX^n_{AB}(\cc,\mm)$ be a random variable corresponding
to the distribution
\begin{equation}
  \mu_{\tX^n}(\xx)
  \equiv
  \begin{cases}
    \frac{\mu_{X^n}(\xx)}{\mu_{X^n}(\C_{AB}(\cc,\mm))},
    &\text{if}\ \xx\in\C_{AB}(\cc,\mm),
    \\
    0,
    &\text{if}\ \xx\notin\C_{AB}(\cc,\mm),
  \end{cases}
  \label{eq:channel-stochastic-encoder}
\end{equation}
where $\mu_{X^n}$ is the probability distribution of the channel input
random variable $X^n$.

We define the stochastic encoder $\Phi_n:\im\B\to\X^n$
and the decoder $\psi_n:\Y^n\to\im\B$ as
\begin{align}
  \Phi_n(\mm)
  &\equiv
  \tX^n_{AB}(\cc,\mm)
  \label{eq:channel-encoder}
  \\
  \psi_n(\yy)
  &\equiv
  B\xx_A(\cc|\yy),
  \label{eq:channel-decoder}
\end{align}
where we declare an encoding error when $\mu_{X^n}(\C_{AB}(\cc,\mm))=0$
and $\xx_A$ is defined as
\begin{equation}
  \xx_{A}(\cc|\yy)
  \equiv\arg\max_{\xx'\in\C_A(\cc)}\mu_{X^n|Y^n}(\xx'|\yy).
  \label{eq:gA}
\end{equation}
The flow of vectors is illustrated in Fig.\ \ref{fig:channel-code}.
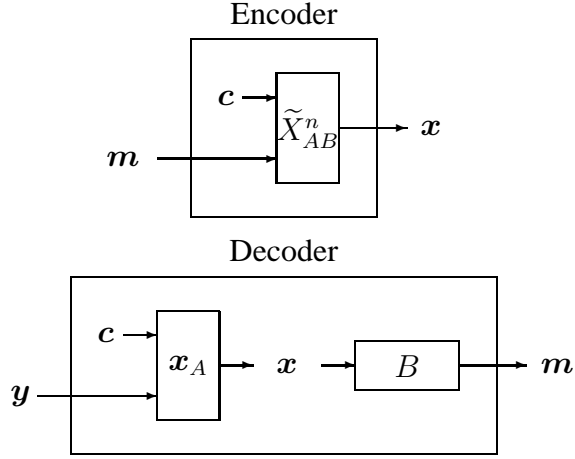
\begin{figure}[t]
  \begin{center}
    \unitlength 0.45mm
    \begin{picture}(176,70)(0,0)
      \put(82,60){\makebox(0,0){Encoder}}
      \put(65,35){\makebox(0,0){$\cc$}}
      \put(70,35){\vector(1,0){10}}
      \put(35,17){\makebox(0,0){$\mm$}}
      \put(45,17){\vector(1,0){35}}
      \put(80,10){\framebox(18,32){$\tX^n_{AB}$}}
      \put(98,26){\vector(1,0){20}}
      \put(125,26){\makebox(0,0){$\xx$}}
      \put(55,0){\framebox(54,52){}}
    \end{picture}
    \\
    \begin{picture}(176,70)(0,0)
      \put(82,60){\makebox(0,0){Decoder}}
      \put(30,35){\makebox(0,0){$\cc$}}
      \put(35,35){\vector(1,0){10}}
      \put(5,17){\makebox(0,0){$\yy$}}
      \put(10,17){\vector(1,0){35}}
      \put(45,10){\framebox(18,32){$\xx_A$}}
      \put(63,26){\vector(1,0){10}}
      \put(83,26){\makebox(0,0){$\xx$}}
      \put(93,26){\vector(1,0){10}}
      \put(103,19){\framebox(30,14){$B$}}
      \put(133,26){\vector(1,0){20}}
      \put(162,26){\makebox(0,0){$\mm$}}
      \put(20,0){\framebox(124,52){}}
    \end{picture}
  \end{center}
  \caption{Construction of Channel Code}
  \label{fig:channel-code}
\end{figure}
\begin{rem}
It should be noted that (\ref{eq:channel-encoder}) is different from
the encoder defined in \cite{HASH} whereas (\ref{eq:channel-decoder}) is
the same.
In \cite{HASH}, the encoder is defined based on typical sets,
where $\xx\in\T_{X|Y,2\gamma}(\yy)$ is satisfied when
$\xx\in\T_{X,\gamma}$ and $\yy\in\T_{Y|X,\gamma}(\xx)$.
We changed the definition of the encoder because a general channel may
not satisfy this property.
\end{rem}

The error probability $\Error(A,B,\cc)$ is given by
\begin{equation}
  \Error(A,B,\cc)
  \equiv
  \sum_{\substack{
      \mm:\\
      \mu_{X^n}(\C_{AB}(\cc,\mm))=0
  }}
  \frac1
  {|\M_n|}
  +
  \sum_{\substack{
      \mm,\xx,\yy:\\
      \mu_{X^n}(\C_{AB}(\cc,\mm))>0\\
      \xx\in\C_{AB}(\cc,\mm)\\
      \psi_n(\yy)\neq\mm
  }}
  \frac{\mu_{Y^n|X^n}(\yy|\xx)\mu_{X^n}(\xx)}
  {|\M_n|\mu_{X^n}(\C_{AB}(\cc,\mm))}.
  \label{eq:channel-error-stochastic}
\end{equation}
We have the following theorem, where the proof is given in
Section~\ref{sec:proof-channel}.
\begin{thm}
\label{thm:channel}
If $r,R>0$ satisfy
\begin{align}
  r&>\oH(\XX|\YY)
  \label{eq:channel-r}
  \\
  r+R&<\uH(\XX),
  \label{eq:channel-rR}
\end{align}
then for any $\delta>0$ and all sufficiently large $n$ there are
functions $A\in\A$, $B\in\B$, and a vector $\cc\in\im\A$ such that
\begin{equation}
  \Error(A,B,\cc)\leq\delta.
  \label{eq:channel-error}
\end{equation}
The channel capacity is achievable with the proposed code by letting
$\XX$ be a source that attains the supremum on the right hand side of
(\ref{eq:capacity}).
\end{thm}

\begin{rem}
From (\ref{eq:channel-error-stochastic}) and (\ref{eq:channel-error}),
we have the fact that $\C_{AB}(\cc,\mm)\neq\emptyset$ with probability
close to $1$ by letting $\delta\to 0$ because
\begin{align}
  \sum_{\substack{
      \mm:\\
      \C_{AB}(\cc,\mm)=\emptyset
  }}
  \frac 1{|\M_n|}
  &\leq
  \sum_{\substack{
      \mm:\\
      \mu_{X^n}(\C_{AB}(\cc,\mm))=0
  }}
  \frac 1{|\M_n|}
  \notag
  \\
  &\leq
  \Error(A,B,\cc)
  \notag
  \\
  &\leq
  \delta.
\end{align}
Furthermore, we can find $\cc\in\im A\subset\im\A$ satisfying
(\ref{eq:channel-error})
because $\Error(A,B,\cc)=1$ when $\cc\in\im\A\setminus\im A$.
\end{rem}

Next, we consider a special case for the proposed code,
which provides an interpretation of
the conventional linear codes \cite{BB04}\cite{GA68}.
It should be noted that a constrained random number generator
is unnecessary.

Let us assume that
$\mu_{X^n}$ is the uniform distribution on $\X^n$
and
$(\bcA,\pA)$ is an ensemble of {\em matrices}
$A:\X^n\to\X^l$
satisfying
\[
  r=\frac 1n\log|\im\A|
\]
when $0<r<\log|\X|$ is given.
We fix a matrix $A\in\A$ and a vector $\cc\in\im A\subset\im\A$
so that they are available for constructing an encoder and a decoder.

Let $\M_n$ be a set of all messages that is a linear space satisfying
$|\M_n|=|\C_A(\cc)|$ for all $\cc\in\im A$.
Since $A$ is a linear function,
there is a bijective linear function $G:\M_n\to\C_A(\zero)$, which is
known as a generator matrix.
The rate $R$ of the code is given as
\[
  R \equiv \frac1n\log|\M_n|.
\]
Since for a given $\cc\in\im A$ there is $\xx_{\cc}$  such that
$A\xx_{\cc}=\cc$,
then we have the fact that $A[G\mm+\xx_{\cc}]=\cc$ for all $\mm\in\M_n$.
Since $G$ is a linear function, there is a linear function $B:\X^n\to\M_n$ 
such that $BG\mm=\mm$ for all $\mm\in\M_n$.
We define a deterministic encoder $\vphi_n:\M_n\to\C_A(\cc)$ and a
decoder $\psi_n:\Y^n\to\M_n$ as
\begin{align*}
  \vphi_n(\mm)
  &\equiv G\mm+\xx_{\cc}
  \\
  \psi_n(\yy)
  &\equiv B[\xx_A(\cc|\yy)-\xx_{\cc}]
\end{align*}
where $\xx_A$ is defined as (\ref{eq:gA}).
The error probability $\Error(A,\cc)$ is given by
\[
  \Error(A,\cc)
  \equiv
  \sum_{\substack{
      \xx\in\C_A(\cc),\yy:\\
      \psi_n(\yy)\neq\xx
  }}
  \frac{\mu_{Y^n|X^n}(\yy|\xx)}
  {|\C_A(\cc)|}.
\]
We have the following corollary,
which is shown in Section~\ref{sec:proof-channel-uniform}.
\begin{cor}
\label{thm:channel-uniform}
If $r$ satisfies
\begin{equation}
  \oH(\XX|\YY)<r<\log|\X|,
  \label{eq:channel-uniform-r}
\end{equation}
then for any $\delta>0$ and all sufficiently large $n$ there are a
matrix $A\in\A$ and a vector $\cc\in\im A$ such that
\begin{gather}
  R\geq \log|\X|-r
  \label{eq:channel-uniform-R}
  \\
  \Error(A,\cc)\leq\delta
  \label{eq:channel-uniform-error}
\end{gather}
When the supremum on the right hand side of (\ref{eq:capacity})
is achieved by $\XX$ corresponding to the uniform distribution,
the channel capacity
\begin{equation}
  C(\WW)=\log|\X|-\oH(\XX|\YY)
  \label{eq:capacity-uniform}
\end{equation}
is achievable with the proposed code by letting $r\to\oH(\XX|\YY)$.
Assuming that $\X=\Y=\Z$ is a finite field, the capacity
\begin{equation}
  C(\WW)=\log|\X|-\oH(\ZZ)
  \label{eq:capacity-additive}
\end{equation}
for a channel with additive noise $\ZZ=\{Y^n-X^n\}_{n=1}^{\infty}$
is achievable with the proposed code by letting $r\to\oH(\ZZ)$.
\end{cor}

\begin{rem}
In~\cite[Thoerem 7.2.1]{CT}, the capacity of 
a discrete stationary memoryless weakly symmetric channel is given as
\[
  C(\WW)=\log|\Y|-H(\text{row of the transition matrix}),
\]
which is another expression of (\ref{eq:capacity-uniform}).
It should be noted that the formula (\ref{eq:capacity-uniform})
is valid for a weakly symmetric channel and is well-defined as long as
$|\X|$ is finite.
It should also be noted that the capacity (\ref{eq:capacity-uniform})
for a symmetric output channel (e.g. an additive Gaussian noise channel)
is achieved by $\XX$ corresponding to the uniform distribution.
For a channel with additive noise $\ZZ$, the channel capacity
(\ref{eq:capacity-additive}) is derived
in~\cite{VH94}\cite[Example 3.2.1]{HAN} when $\X=\Y=\Z=\{0,1\}$.
Formula (\ref{eq:capacity-additive}) is an extension to a general finite
alphabet.
\end{rem}

\section{Construction of Lossy Source Code}
\label{sec:lossy}

This section introduces a lossy source code.
The idea for the construction is drawn
from~\cite{SWLDPC}\cite{HASH}\cite{ISIT2011a}.
It should be noted that we assume that a reproduction alphabet $\X^n$ is
finite set but a source alphabet $\Y^n$ is allowed to be arbitrary
(infinite, continuous) set.

For given $r>$ and $R>0$,
let $(\bcA,\pA)$ and $(\bcB,\pB)$ be ensembles of functions
$A:\X^n\to\im\A$ and $B:\X^n\to\im\B$ satisfying
\begin{align*}
  r&=\frac1n\log|\im\A|
  \\
  R&=\frac1n\log|\im\B|,
\end{align*}
respectively.
We fix functions $A\in\A$, $B\in\B$, and a vector $\cc\in\im\A$
so that they are available for constructing an encoder and a decoder.

Let  $\mu_{X^n}$ be defined as
\[
  \mu_{X^n}(\xx)
  \equiv
  \sum_{\yy}\mu_{X^n|Y^n}(\xx|\yy)\mu_{Y^n}(\yy),
\]
where $\mu_{Y^n}$ is the probability distribution of a source $Y^n$
and we assume that the conditional probability distribution
$\mu_{X^n|Y^n}$ is given.
We use a constrained random number generator to construct an encoder.
Let $\tX^n\equiv\tX^n_A(\cc|\yy)$ be a random variable corresponding to
the distribution
\begin{equation}
  \mu_{\tX^n|Y^n}(\xx|\yy)\equiv
  \begin{cases}
    \frac{\mu_{X^n|Y^n}(\xx|\yy)}{\mu_{X^n|Y^n}(\C_A(\cc)|\yy)},
    &\text{if}\ \xx\in\C_A(\cc),
    \\
    0,
    &\text{if}\ \xx\notin\C_A(\cc).
  \end{cases}
  \label{eq:lossy-stochastic-encoder}
\end{equation}
We define the stochastic encoder $\Phi_n:\Y^n\to\im\B$ and the
decoder $\psi_n:\im\B\to\X^n$ as
\begin{align}
  \Phi_n(\yy)
  &\equiv B\tX^n_A(\cc|\yy)
  \label{eq:lossy-encoder}
  \\
  \psi_n(\mm)
  &\equiv \xx_{AB}(\cc,\mm),
  \label{eq:lossy-decoder}
\end{align}
where we declare an encoding error when $\mu_{X^n|Y^n}(\C_A(\cc)|\yy)=0$
and $\xx_{AB}$ is defined as
\begin{align}
  \xx_{AB}(\cc,\mm)
  &\equiv
  \arg\max_{\xx'\in\C_{AB}(\cc,\mm)}
  \mu_{X^n}(\xx').
  \label{eq:lossy-gAB}
\end{align}
The flow of vectors is illustrated in Fig.\ \ref{fig:lossy}.
\begin{figure}[t]
  \begin{center}
    \unitlength 0.45mm
    \begin{picture}(176,70)(0,0)
      \put(82,60){\makebox(0,0){Encoder}}
      \put(30,35){\makebox(0,0){$\cc$}}
      \put(35,35){\vector(1,0){10}}
      \put(5,17){\makebox(0,0){$\yy$}}
      \put(10,17){\vector(1,0){35}}
      \put(45,10){\framebox(18,32){$\tX^n_{A}$}}
      \put(63,26){\vector(1,0){10}}
      \put(83,26){\makebox(0,0){$\xx$}}
      \put(93,26){\vector(1,0){10}}
      \put(103,19){\framebox(30,14){$B$}}
      \put(133,26){\vector(1,0){20}}
      \put(162,26){\makebox(0,0){$\mm$}}
      \put(20,0){\framebox(124,52){}}
    \end{picture}
    \\
    \begin{picture}(176,70)(0,0)
      \put(82,60){\makebox(0,0){Decoder}}
      \put(65,35){\makebox(0,0){$\cc$}}
      \put(70,35){\vector(1,0){10}}
      \put(35,17){\makebox(0,0){$\mm$}}
      \put(45,17){\vector(1,0){35}}
      \put(80,10){\framebox(18,32){$\xx_{AB}$}}
      \put(98,26){\vector(1,0){20}}
      \put(125,26){\makebox(0,0){$\xx$}}
      \put(55,0){\framebox(54,52){}}
    \end{picture}
  \end{center}
  \caption{Construction of Lossy Source Code}
  \label{fig:lossy}
\end{figure}
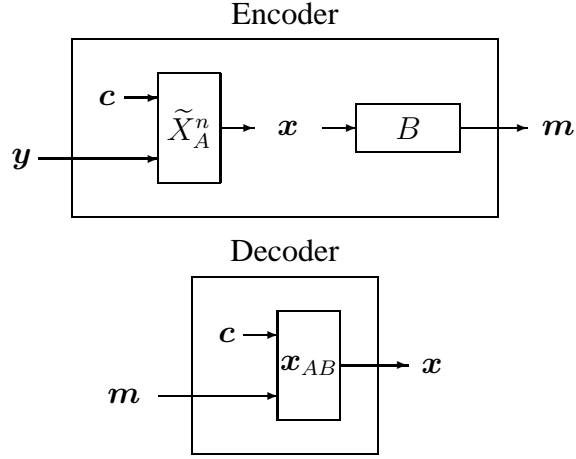

The error probability $\Error(A,B,\cc,D)$ is given as
\begin{equation}
  \Error(A,B,\cc,D)
  \equiv
  P\lrsb{d_n(\psi_n(\Phi_n(Y^n)),Y^n)>D},
  \label{eq:lossy-error}
\end{equation}
where we define $d_n(\psi_n(\Phi_n(\yy)),\yy)\equiv\infty$
when $\mu_{X^n|Y^n}(\C_A(\cc)|\yy)=0$.
We have the following theorem,
where the proof is given in Section~\ref{sec:proof-lossy}.
\begin{thm}
\label{thm:lossy}
If $r,R>0$ satisfy
\begin{align}
  r
  &<
  \uH(\XX|\YY)
  \label{eq:lossy-r}
  \\
  r+R
  &>
  \oH(\XX),
  \label{eq:lossy-rR}
\end{align}
then for any $\delta>0$
and all sufficiently large $n$
there are functions $A\in\A$, $B\in\B$, and
a vector $\cc\in\im\A$ such that
\begin{equation}
  \Error(A,B,\cc,D)\leq
  P\lrsb{d_n(X^n,Y^n)>D}+\delta.
  \label{eq:lossy-error-bound}
\end{equation}
\end{thm}

By assuming that $\{\mu_{X^n|Y^n}\}_{n=1}^{\infty}$ satisfies
\[
  \oD(\XX,\YY)<D,
\]
we have the fact that
\[
  \limn P\lrsb{d_n(X^n,Y^n)>D}=0
\]
from the definition of $\oD(\XX,\YY)$.
Then, by letting $n\to\infty$, $\delta\to0$, and $r\to\uH(\XX|\YY)$, 
we have the fact that for any $(R,D)$ close to the boundary of $\R(\YY)$
there is a sequence of proposed codes such that
\[
  \limn \Error(A,B,\cc,D)=0
\]

\begin{rem}
We can find $\cc\in\im A\subset\im\A$ satisfying
(\ref{eq:lossy-error-bound}) because $\Error(A,B,\cc)=1$ when
$\C_A(\cc)=\emptyset$.
\end{rem}

Next, we consider a special case of the proposed code,
which provides an interpretation of the conventional code introduced
in~\cite{MY03}\cite{GV09}.
It should be noted that a constrained random number generator
is unnecessary.

Let us assume that $\mu_{X^n}$ is the uniform distribution on $\X^n$
and $(\bcA,\pA)$ is an ensemble of {\em matrices} $A:\X^n\to\X^l$
satisfying
\[
  r=\frac 1n\log|\im\A|
\]
when $r>0$ is given.
We fix a matrix $A\in\A$ and a vector $\cc\in\im A\subset\im\A$
so that they are available for constructing an encoder and a decoder.

Since $\C_A(\zero)$ is a linear space,
there is a surjective linear function $B:\X^n\to\C_A(\zero)$.
We use the encoder defined by (\ref{eq:lossy-encoder}).
The rate $R$ of the code is given as
\[
  R \equiv \frac 1n\log|\C_A(\zero)|.
\]
Furthermore, since $B$ is surjective, there is a bijective linear
function $\xx'_{AB}:\im A\times\C_A(\zero)\to\X^n$ such that
$\xx'_{AB}(A\xx,B\xx)=\xx$ for all $\xx$.
We replace the function $\xx_{AB}$ by $\xx'_{AB}$ in the definition
of the decoder (\ref{eq:lossy-decoder}).
Let $\Error(A,\cc,D)$ be the error probability given as
\[
  \Error(A,\cc,D)
  \equiv
  P\lrsb{d_n(\psi_n(\Phi_n(Y^n)),Y^n)>D}.
\]
We have the following corollary,
which is shown in Section~\ref{sec:proof-lossy-uniform}.
\begin{cor}
\label{thm:lossy-uniform}
If $r$ satisfies
\begin{equation}
  r<\uH(\XX|\YY),
  \label{eq:lossy-uniform-r}
\end{equation}
then for any $\delta>0$ and all sufficiently large $n$
there are a matrix $A\in\A$ and a vector $\cc\in\im\A$ such that
\begin{align}
  R
  &\leq
  \log|\X|-r+\delta
  \\
  \Error(A,\cc,D)
  &\leq
  P\lrsb{d(X^n,Y^n)>D}+\delta.
  \label{eq:lossy-uniform-d}
\end{align}
When the boundary of the right hand side of (\ref{eq:rd}) is attained
with a general channel $\WW\equiv\{\mu_{X^n|Y^n}\}_{n=1}^{\infty}$
such that $\mu_{X^n}$ is uniform,
for any $(R,D)$ close to the boundary of $\R(\YY)$ there is a sequence
of proposed codes such that
\[
  \limn \Error(A,\cc,D)=0.
\]
\end{cor}

\section{Constrained Random Number Generation by Using Sum-Product
  Algorithm}
\label{sec:sum-product}

In this section, we introduce an algorithm
for generating random numbers subject to the distributions
(\ref{eq:channel-stochastic-encoder}) and
(\ref{eq:lossy-stochastic-encoder})
by assuming that $\mu_{X^n}$ and $\mu_{X^n|Y^n}$ are memoryless,
that is, they are given by
\begin{align}
  \mu_{X^n}(\xx)&=\prod_{i=1}^n\mu_{X_i}(x_i)
  \label{eq:sp-muX}
  \\
  \mu_{X^n|Y^n}(\xx|\yy)&=\prod_{i=1}^n\mu_{X_i|Y_i}(x_i|y_i)
  \notag
\end{align}
for each $\xx\equiv(x_1,\ldots,x_n)$ and $\yy\equiv(y_1,\ldots,y_n)$,
respectively.
In the following,
we construct a random number generator subject
to the distribution $\mu_{\tX^n}$
defined by
\begin{equation}
  \mu_{\tX^n}(\xx)\equiv
  \begin{cases}
    \frac{\mu_{X^n}(\xx)}{\mu_{X^n}(C_A(\cc))}
    &\text{if}\ \xx\in\C_A(\cc)
    \\
    0
    &\text{if}\ \xx\notin\C_A(\cc)
  \end{cases}
  \label{eq:sum-product-tX}
\end{equation}
for a $\mu_{X^n}$ given  by (\ref{eq:sp-muX}),
$A$, and $\cc\in\im\A$.
It should be noted that (\ref{eq:channel-stochastic-encoder})
can be reduced to (\ref{eq:sum-product-tX}) by considering a function
$(A,B):\X^n\to\im\A\times\im\B$ defined as $(A,B)\xx\equiv(A\xx,B\xx)$,
and (\ref{eq:lossy-stochastic-encoder}) can also be reduced to
(\ref{eq:sum-product-tX}) by letting
$\mu_{X_i}\equiv\mu_{X_i|Y_i}(\cdot|y_i)$ for a given $\yy$.

Let us assume that there is a family $\{\Z_j\}_{j\in\J}$ of sets
such that $\im\A\subset\times_{j\in\J}\Z_j$.
For a set of local functions $\{f_j:\X^{|\cS_j|}\to\Z_j\}_{j\in\J}$,
the sum-product algorithm~\cite{GDL}\cite{KFL01} calculates a
real-valued global function $g$ on $\X$ defined as
\[
  g(x_i)
  \equiv
  \frac{
    \sum_{\xx\setminus\{x_i\}}
    \prod_{j\in\J}f_j(\xx_{\cS_j})
  }{
    \sum_{\xx}\prod_{j\in\J}f_j(\xx_{\cS_j})
  }
\]
approximately, where the summention $\sum_{\xx\setminus\{x_i\}}$
is taken over all $\xx\in\X^n$ except for the variable $x_i$ and the
function $f_j$ depends only on the set of variables
$\xx_{\cS_i}\equiv\{x_j\}_{j\in\cS_i}$.
It should be noted that the algorithm calculates
the global function exactly when the corresponding factor graph has no loop.
Let $\pi_{x_i\to f_j}(x_i)$ and $\sigma_{f_j\to x_i}(x_i)$
be messages defined as
\begin{align*}
  \pi_{x_i\to f_j}(x_i)
  &\equiv
  \prod_{j'\in\J\setminus\{j\}: i\in\cS_{j'}}
  \sigma_{f_{j'}\to x_i}(x_i)
  \\
  \sigma_{f_j\to x_i}(x_i)
  &\equiv
  \frac{
    \sum_{\xx_{\cS_j\setminus\{i\}}}
    f_j(\xx_{\cS_j})
    \prod_{i'\in\cS_{j}\setminus\{i\}}
    \pi_{x_{i'}\to f_j}(x_{i'})
  }{
    \sum_{\xx_{\cS_j}}
    f_j(\xx_{\cS_j})
    \prod_{i'\in\cS_{j}\setminus\{i\}}
    \pi_{x_{i'}\to f_j}(x_{i'})
  },
\end{align*}
where the summation  $\sum_{\cS}$ is taken over all $\{x_i\}_{i\in\cS}$,
$\pi_{x_i\to f_j}(x_i)\equiv 1$ when there is no $j'\in\J\setminus\{j\}$
such that $i\in\cS_{j'}$ and
$\sigma_{f_j\to x_i}(x_i)\equiv f_j(x_i)/\sum_{x_i}f_j(x_i)$
when $\cS_j=\{i\}$.
The sum-product algorithm is performed by repeating the above operations
for every message $\sigma_{f_j\to x_i}(x_i)$ and $\pi_{x_i\to f_j}(x_i)$
satisfying $i\in\cS_j$ and finally calculating the approximation of the
global function as
\[
  g(x_i)
  \approx
  \prod_{j\in\J: i\in\cS_j}
  \sigma_{f_j\to x_i}(x_i),
\]
where we assign initial values to $\pi_{x_i\to f_j}(x_i)$ and
$\sigma_{f_j\to x_i}(x_i)$ when they appear on the right hand side of
the above operations and are undefined.

In the following, we introduce an algorithm for constrained random
number generation.
For each $i\in\{1,\ldots,l\}$,
let $\ba_i:\X^{|\cS_i|}\to\Z$ be a function such that
\[
  A\xx=(\ba_1(\xx_{\cS_1}),\ba_2(\xx_{\cS_2}),\ldots,\ba_l(\xx_{\cS_l})),
\]
where the $i$-th component $\ba_i$ of $A$ depends only on the
set of variables $\xx_{\cS_i}\equiv\{x_j\}_{j\in\cS_i}$.
For example, when $A\equiv(a_{i,j})$ is an $l\times n$ sparse matrix
with a maximum row weight $w$, we have $\X=\Z$, the set $\cS_i$ defined
as
\[
  \cS_i\equiv\{j\in\{1,\ldots,n\}: a_{i,j}\neq 0\}
\]
satisfies $|\cS_i|\leq w$ for all $i\in\{1,\ldots,l\}$,
and $\ba_i(\xx_{\cS_i})$ is defined as the inner product $\ba_i\cdot\xx$
of vectors $\ba_i$ and $\xx$.
Let $x_i^j\equiv(x_i,\ldots,x_j)$, where $x_i^j$ is a null string
if $i>j$. Let $\cc\equiv(c_1,\ldots,c_l)\in\Z^l$.
Let $\chi(\cdot)$ be defined as
\begin{align}
  \chi(S)
  &\equiv
  \begin{cases}
    1,&\text{if the statement $S$ is true}
    \\
    0,&\text{if the statement $S$ is false}.
  \end{cases}
  \label{eq:chi}
\end{align}

\noindent{\bf Constrained Random Number Generation Algorithm:}
\begin{algorithm}{Step 99}{}
  \item Let $k=1$. 
  \item Calculate the conditional probability distribution
  $p_{\tX_k|\tX_1^{k-1}}$
  defined as
  \begin{equation}
    p_{\tX_k|\tX_1^{k-1}}(x_k|x_1^{k-1})
    \equiv
    \frac{
      \sum_{x_{k+1}^n}
      \prod_{j=k}^n
      \mu_{X_j}(x_j)
      \prod_{i=1}^{l}\chi(\ba_i(\xx_{\cS_i})=c_i)
    }{
      \sum_{x_k^n}
      \prod_{j=k}^n
      \mu_{X_j}(x_j)
      \prod_{i=1}^{l}\chi(\ba_i(\xx_{\cS_i})=c_i)
    }.
    \label{eq:sum-product-fk}
  \end{equation}
  It should be noted that the sum-product algorithm can be employed
  to obtain (\ref{eq:sum-product-fk}),
  where $\{\mu_{X_j}\}_{j=k}^n$ and
  $\{\chi(\ba_i(\xx_{\cS_i})=c_i)\}_{i=1}^l$ are local functions and
  we substitute the generated sequence $x_1^{k-1}$ for
  (\ref{eq:sum-product-fk}).
  If $\chi(\ba_i(\xx_{\cS_i})=c_i)$ is a constant after the substitution
  of $x_1^{k-1}$,
  we can recode the constant in preparation for the future.
  \item Generate and recode a random number $x_k$ corresponding to the
  distribution $p_{\tX_k|\tX_1^{k-1}}$.
  \item If $k=n$, output $\xx\equiv x_1^n$ and terminate.
  \item If for the generated sequence $x_1^k$ there is a unique
  $x_{k+1}^n$ such that $\xx\equiv x_1^n\in\C_A(\xx)$,
  obtain the unique vector $x_{k+1}^n$, output $\xx$, and terminate.
  \item
  Let $k\leftarrow k+1$ and go to {\sf Step 2}.
\end{algorithm}

\begin{rem}
We can omit {\sf Step 5} if it is hard to execute.
\end{rem}
\begin{rem}
When $A$ is a linear function with rank $l'$,
by checking whether $k=l'$ or not at {\sf Step 5},
we can easily determine whether or not for a given $x_1^k$ there is a
unique $x_{k+1}^n$ such that $x_1^n\in\C_A(\xx)$.
We can obtain the unique $x_{l'+1}^n$ from
$x_1^{l'}$ by using a linear operation.
\end{rem}
\begin{rem}
It should be noted that the memoryless condition on $X^n$ is not
essential for the description of the algorithm.
The algorithm is well-defined when we use the formula
\[
  \mu_{X^n}(\xx)=\prod_{i=1}^n\mu_{X_i|X_1^{i-1}}(x_i|x_1^{i-1})
\]
and replace $\mu_{X_i}(x_i)$ by $\mu_{X_i|X_1^{i-1}}(x_i|x_1^{i-1})$
for $i\geq 2$ in (\ref{eq:sum-product-fk}).
However, the sum-product algorithm may not find a good approximation
in general because the corresponding factor graph may have many loops.
\end{rem}

We have the following theorem, which is shown in
Section~\ref{sec:proof-sum-product}.
\begin{thm}
\label{thm:sum-product}
Assume that (\ref{eq:sum-product-fk}) is computed exactly.
Then the proposed algorithm generates $\xx\equiv x_1^n$
subject to the probability distribution given by
(\ref{eq:sum-product-tX}).
\end{thm}

In the following, we consider a situation where we can use a real number
$\oomega$ subject to the uniform distribution on $[0,1)$.
We modify the proposed algorithm,
where the basic idea comes from the interval algorithm
introduced in~\cite{HH97} and is analogous to the arithmetic
coding~\cite{RL76}.
It should be noted that only {\sf Steps 1, 3} are modified.

\noindent{\bf Interval Constrained Random Number Generation Algorithm:}
\begin{algorithm}{Step 99}{}
  \item Let $k=1$ and $[\utheta_1,\otheta_1)\equiv[0,1)$.
  \item Calculate the conditional probability distribution
  $p_{\tX_k|\tX_1^{k-1}}$ defined by (\ref{eq:sum-product-fk}).
  \item Partition the interval $[\utheta_{k-1},\otheta_{k-1})$ 
  into sub-intervals that are labeled corresponding to the elements in $\X$,
  where the sub-interval width is subject to the ratio
  $p_{\tX_k|\tX_1^{k-1}}(x_k|x_1^{k-1})$.
  Let $[\utheta_k,\otheta_k)$ be a sub-interval
  that contains $\oomega$, that is, $\oomega\in[\utheta_k,\otheta_k)$ is
  satisfied for a given $\oomega$.
  Let $x_k$ be a label that corresponds to the sub-interval 
  $[\utheta_k,\otheta_k)$ and record it.
  \item If $k=n$, output $\xx\equiv x_1^n$ and terminate.
  \item If for the generated sequence $x_1^k$ there is a unique
  $x_{k+1}^n$ such that $\xx\equiv x_1^n\in\C_A(\xx)$,
  obtain the unique vector $x_{k+1}^n$, output $\xx$, and terminate.
  \item Let $k\leftarrow k+1$ and go to {\sf Step 2}.
\end{algorithm}

From Theorem~\ref{thm:sum-product},
we have the fact that the probability of selecting $\oomega\in[0,1)$ is
equal to the width of the sub-interval $[\utheta_{k'},\otheta_{k'})$,
which is equal to the probability $\mu_{\tX^n}(\xx)$ of a generated
sequence $\xx$,
where $k'$ is the value of $k$ when the algorithm is terminated.

It should be noted that we can construct a deterministic code
from a stochastic code by fixing a random number $\oomega\in[0,1)$.
In fact, by using the random coding argument,
we can show that there is a random number $\oomega\in[0,1)$
such that the error probability is sufficiently small.
This is because, from Theorems~\ref{thm:channel} and \ref{thm:lossy},
the average error probability with respect to the random variable
corresponding to a random number on $[0,1)$ is sufficiently small.

\begin{rem}
Instead of a real number $\oomega$,
we can use a binary random sequence $\omega_1,\omega_2,\ldots$
subject to the uniform distribution on $\{0,1\}$
by letting $\oomega\equiv 0.\omega_1\omega_2\cdots\in[0,1)$,
which is the binary expansion of a real number.
Since we can estimate $\mu_{X^n}(\C_A(\cc))=1/|\im\A|$  approximately
and the average entropy of $\tX^n(\cc)$ is given as
\begin{align}
  E_{\sfcc}\lrB{H(\tX^n(\sfcc))}
  &=
  \sum_{\cc}\frac 1{|\im\A|}\sum_{\xx\in\C_A(\cc)}\frac{\mu_{X^n}(\xx)}{\mu_{X^n}(\C_A(\cc))}
  \log\frac{\mu_{X^n}(\C_A(\cc))}{\mu_{X^n}(\xx)}
  \notag
  \\
  &=
  H(X^n)-\log|\im\A|,
\end{align}
the required length of the binary sequence can be estimated
approximately as at least $H(X^n)-\log|\im\A|$.
\end{rem}

\section{Proofs of Theorems}
\label{sec:proof}

\subsection{Proof of Lemma~\ref{lem:capacity}}
\label{sec:proof-capacity}
Since
\[
  \uI(\XX;\YY)\geq\uH(\XX)-\oH(\XX|\YY)
\]
for any $(\XX,\YY)$, we have
\begin{align}
  C(\WW)
  &=
  \sup_{\XX}\uI(\XX;\YY)
  \notag
  \\*
  &\geq 
  \sup_{\XX}\lrB{\uH(\XX)-\oH(\XX|\YY)}.
\end{align}
In the following, we prove that
\begin{equation}
  C(\WW)\leq\sup_{\XX}\lrB{\uH(\XX)-\oH(\XX|\YY)},
  \label{eq:proof-capacity-leq}
\end{equation}
which completes the proof of the lemma.

From the definition of $C(\WW)$,
we have the fact that for any $\delta>0$ and sufficiently large $n$
there is a pair consisting an encoder $\vphi_n:\M_n\to\X^n$ and 
a decoder $\psi_n:\Y^n\to\M_n$ such that
\begin{align}
  \liminfn\frac 1n\log|\M_n| \geq C(\WW)-\delta
  \label{eq:proof-capacity-rate}
  \\
  \limn P(\psi_n(Y^n)\neq M_n)=0.
  \label{eq:proof-capacity-error}
\end{align}
We can assume\footnote{This assumption is used merely so that
  $\MM\equiv\{M_n\}_{n=1}^{\infty}$ is a general source satisfying
  $M_n\in\X^n$.
  It should be noted that $\M_n$ and $\{\vphi_n(\mm):\mm\in\M_n\}$ are
  different subsets of $\X^n$ in general.
  We could define a channel code by a subset $\M_n$ of $\X^n$ as defined
  in~\cite{HAN}\cite{VH94} instead of introducing an encoder $\vphi_n$.
  We introduce an encoder $\vphi_n$ to consider a stochastic encoder.}
that $\M_n\subset\X^n$ without loss of generality.
Since the distribution $\mu_{M_n}$ of $M_n$ is uniform on $\M_n$,
we have the fact that
\begin{align}
  \frac1n\log\frac 1{\mu_{M_n}(\xx)}
  &=
  \frac1n\log|\M_n|
  \notag
  \\
  &\geq
  \liminfn\frac 1n \log|\M_n|-\delta
\end{align}
for all $\xx\in\M_n$, $\delta>0$, and sufficiently large $n$.
Since 
\[
  \frac1n\log\frac 1{\mu_{M_n}(\xx)}=\infty
\]
for every $\xx\notin\M_n$, we have the fact that
\[
  \frac1n\log\frac 1{\mu_{M_n}(\xx)}\geq 
  \liminfn\frac 1n \log|\M_n|-\delta
\]
for every $\xx\in\X^n$, $\delta>0$ and sufficiently large $n$.
This implies that
\begin{equation}
  \limn
  \Prob\lrsb{
    \frac1n\log\frac 1{\mu_{M_n}(M_n)}
    <\liminfn\frac 1n \log|\M_n|-\delta
  }
  =
  0.
  \label{eq:proof-capacity-PM}
\end{equation}
Let $\MM\equiv\{M_n\}_{n=1}^{\infty}$ be a general source.
Then we have 
\begin{equation}
  \liminfn\frac 1n \log|\M_n|-\delta\leq \uH(\MM)
  \label{eq:proof-capacity-uHM}
\end{equation}
from (\ref{eq:proof-capacity-PM}) and the definition of $\uH(\MM)$.
We have
\begin{align}
  C(\WW)
  &\leq
  \liminfn\frac1n\log|\M_n|+\delta
  \notag
  \\*
  &\leq
  \uH(\MM)+2\delta
  \notag
  \\
  &=
  \uH(\MM)-\oH(\MM|\YY)+2\delta
  \notag
  \\
  &\leq
  \sup_{\XX}\lrB{
    \uH(\XX)-\oH(\XX|\YY)
  }
  +2\delta,
\end{align}
where the first inequality comes from (\ref{eq:proof-capacity-rate}),
the second inequality comes from (\ref{eq:proof-capacity-uHM}),
and the equality comes from
(\ref{eq:proof-capacity-error}) and Lemma~\ref{lem:fano}.
We have (\ref{eq:proof-capacity-leq}) by letting $\delta\to0$.
\hfill\qed

\subsection{Proof of Lemma~\ref{lem:rd}}
\label{sec:proof-rate-distortion}

Since
\[
  \oI(\XX;\YY)\leq\oH(\XX)-\uH(\XX|\YY)
\]
for any $(\XX,\YY)$, we have
\begin{align*}
  \R(\YY)
  &=
  \bigcup_{\WW}\lrb{
    (R,D):
    \begin{gathered}
      \oI(\XX;\YY)\leq R
      \\
      \oD(\XX;\YY)\leq D
    \end{gathered}
  }
  \notag
  \\
  &\supset
  \bigcup_{\WW}\lrb{
    (R,D):
    \begin{gathered}
      \oH(\XX)-\uH(\XX|\YY)\leq R
      \\
      \oD(\XX;\YY)\leq D
    \end{gathered}
  }.
\end{align*}
In the following, we prove that
\begin{equation}
  \R(\YY)
  \subset
  \bigcup_{\WW}\lrb{
    (R,D):
    \begin{gathered}
      \oH(\XX)-\uH(\XX|\YY)\leq R
      \\
      \oD(\XX;\YY)\leq D
    \end{gathered}
  },
  \label{eq:proof-rd-achievable}
\end{equation}
which completes the proof of the lemma.

Assume that $(R,D)\in\R(\YY)$.
From (\ref{eq:rd-I}),
we have the fact that for all $\delta>0$ and all sufficiently large $n$,
there is a pair consisting an encoder $\vphi_n$ and a decoder $\psi_n$
satisfying (\ref{eq:rd-R}) and (\ref{eq:rd-D}).
Let $\hX^n\equiv\psi_n(\vphi_n(Y^n))\in\X^n$.
Then we have
\begin{align}
  \Prob\lrsb{
    \frac1n\log\frac1{\mu_{\hX^n}(\hX^n)}
    >
    R+\e
  }
  &\leq
  \Prob\lrsb{
    \frac1n\log\frac1{\mu_{\hX^n}(\hX^n)}
    \geq
    \frac1n\log|\M_n|+\e
  }
  \notag
  \\
  &\leq
  2^{-n\e}
\end{align}
for any $\e>0$,
where the first inequality comes from (\ref{eq:rd-R}),
and the second inequality comes from \cite[Lemma 2.6.2]{HAN}
and the fact that the cardinality of the domain of $\hX^n$ is at most
$|\M_n|$.
By letting $n\to\infty$, we have the fact that a general source
$\hXX\equiv\{\psi_n(\vphi_n(Y^n))\}_{n=1}^{\infty}$ satisfies
\begin{align}
  \oH(\hXX)-\uH(\hXX|\YY)
  &\leq
  \oH(\hXX)
  \notag
  \\
  &\leq
  R+\e.
\end{align}
By letting $\e\to0$, we have
\begin{equation}
  \oH(\hXX)-\uH(\hXX|\YY)\leq R.
  \label{eq:proof-rd-hX-R}
\end{equation}
On the other hand, we have
\[
  \limn P\lrsb{d_n(\hX^n,Y^n)>D}
  =0
\]
from (\ref{eq:rd-D}) by letting $n\to\infty$ and $\delta\to0$.
This implies that
\begin{equation}
  \oD(\hXX,\YY)
  \leq
  D.
  \label{eq:proof-rd-hX-D}
\end{equation}
Then we have 
\[
  (R,D)
  \in
  \bigcup_{\WW}\lrb{
    (R,D):
    \begin{gathered}
      \oH(\XX)-\uH(\XX|\YY)\leq R
      \\
      \oD(\XX;\YY)\leq D
    \end{gathered}
  },
\]
which implies (\ref{eq:proof-rd-achievable}).
\hfill\qed

\subsection{Proof of Theorem~\ref{thm:channel}}
\label{sec:proof-channel}
We omit dependence on $n$ of $X$ and $Y$ when they appear
in the subscript of $\mu$.

From (\ref{eq:channel-r}) and (\ref{eq:channel-rR}),
we have the fact that there is $\e>0$ satisfying
\begin{align}
  r
  &>
  \oH(\XX|\YY)+\e
  \label{eq:proof-channel-r-e}
  \\
  r+R
  &<
  \uH(\XX)-\e.
  \label{eq:proof-channel-rR-e}
\end{align}

Let $\uT_X\subset\X^n$
and $\oT_{X|Y}\subset\X^n\times\Y^n$
be defined as
\begin{align}
  \uT_X
  &\equiv
  \lrb{
    \xx:
    \frac 1n\log\frac 1{\mu_X(\xx)}
    \geq
    \uH(\XX)-\e
  }
  \\
  \oT_{X|Y}
  &\equiv
  \lrb{
    (\xx,\yy):
    \frac 1n\log\frac 1{\mu_{X|Y}(\xx|\yy)} \leq \oH(\XX|\YY)+\e
  }.
  \label{eq:channel-TXY}
\end{align}
Assume that $(\xx,\yy)\in\oT_{X|Y}$ and $\xx_A(A\xx|\yy)\neq \xx$.
Then we have the fact that there is
$\xx'\in\C_A(A\xx)$ such that $\xx'\neq\xx$ and
\[
  \mu_{X|Y}(\xx'|\yy)
  \geq
  \mu_{X|Y}(\xx|\yy)
  \geq 2^{-n[\oH(\XX|\YY)+\e]}.
\]
This implies that $\lrB{\oT_{X|Y}(\yy)\setminus\{\xx\}}\cap\C_A(A\xx)\neq\emptyset$,
where  $\oT_{X|Y}(\yy)$ is defined as
\[
  \oT_{X|Y}(\yy)\equiv
  \lrb{
    \xx: 
    (\xx,\yy)\in\oT_{X|Y}
  }.
\]
We have
\begin{align}
  E_{\sfA}\lrB{
    \chi(\xx_{\sfA}(\sfA\xx|\yy)\neq\xx)
  }
  &\leq
  p_{\sfA}\lrsb{\lrb{A:
      [\oT_{X|Y}(\yy)\setminus\{\xx\}]\cap\C_A(A\xx)\neq\emptyset
  }}
  \notag
  \\*
  &\leq
  \frac{|\oT_{X|Y}(\yy)|\alphaA{}}
  {|\im\A|}
  +\betaA{}
  \notag
  \\
  &\leq
  2^{-n[r-\oH(\XX|\YY)-\e]}\alphaA{}
  +\betaA{}
  \label{eq:proof-channel-error1-sub}
\end{align}
for all $(\xx,\yy)\in\oT_{X|Y}$,
where $\chi(\cdot)$ is defined by (\ref{eq:chi}),
the second inequality comes from Lemma~\ref{lem:CRP},
and the third inequality comes from the fact that
$|\oT_{X|Y}(\yy)|\leq 2^{n[\oH(\XX|\YY)+\e]}$.
We have the fact that
\begin{align}
  &
  E_{\sfA}\lrB{
    \sum_{\xx,\yy}
    \mu_{XY}(\xx,\yy)
    \chi(\xx_{\sfA}(\sfA\xx|\yy)\neq\xx)
  }
  \notag
  \\*
  &=
  \sum_{(\xx,\yy)\in\oT_{X|Y}}
  \mu_{XY}(\xx,\yy)
  E_{\sfA}\lrB{
    \chi(\xx_{\sfA}(\sfA\xx|\yy)\neq\xx)
  }
  +
  \sum_{(\xx,\yy)\notin\oT_{X|Y}}
  \mu_{XY}(\xx,\yy)
  E_{\sfA}\lrB{
    \chi(\xx_{\sfA}(\sfA\xx|\yy)\neq\xx)
  }
  \notag
  \\
  &\leq
  2^{-n[r-\oH(\XX|\YY)-\e]}\alphaA{}
  +\betaA{}
  +\mu_{XY}([\oT_{X|Y}]^c),
  \label{eq:proof-channel-error1}
\end{align}
where the last inequality comes from (\ref{eq:proof-channel-error1-sub}).
We also have the fact that
\begin{align}
  &
  E_{\sfA\sfB}\lrB{
    \sum_{\cc,\mm}
    \lrbar{
      \mu_{X}(\C_{\sfA\sfB}(\cc,\mm))
      -
      \frac1
      {|\im\A||\im\B|}
    }
  }
  \notag
  \\*
  &\leq
  E_{\sfA\sfB}\lrB{
    \sum_{\cc,\mm}
    \lrbar{
      \mu_{X}(\C_{\sfA\sfB}(\cc,\mm)\cap\uT_X)
      -
      \frac{\mu_X(\uT_X)}
      {|\im\A||\im\B|}
    }
  }
  \notag
  \\*
  &\quad
  +
  E_{\sfA\sfB}\lrB{
    \sum_{\cc,\mm}
    \lrB{
      \mu_{X}(\C_{\sfA\sfB}(\cc,\mm)\cap[\uT_X]^c)
      +
      \frac{\mu_X([\uT_X]^c)}
      {|\im\A||\im\B|}
    }
  }
  \notag
  \\
  &=
  \mu_X(\uT_X)
  E_{\sfA\sfB}\lrB{
    \sum_{\cc,\mm}
    \lrbar{
      \frac{\mu_{X}(\C_{\sfA\sfB}(\cc,\mm)\cap\uT_X)}
      {\mu_X(\uT_X)}
      -
      \frac1
      {|\im\A||\im\B|}
    }
  }
  +
  2\mu_{X}([\uT_X]^c)
  \notag
  \\
  &\leq
  \mu_X(\uT_X)
  \sqrt{
    \frac{\alphaA{}-1+[\betaA{}+1]|\im\A||\im\B|\displaystyle\max_{\xx\in\uT_X}\mu_X(\xx)}
    {\mu_X(\uT_X)}
  }
  +
  2\mu_{X}([\uT_X]^c)
  \notag
  \\*
  &\leq
  \sqrt{\alphaA{}-1+[\betaA{}+1]2^{-n[\uH(\XX)-r-R-\e]}}
  +
  2\mu_{X}([\uT_X]^c),
  \label{eq:proof-channel-error2}
\end{align}
where the second inequality comes from Lemma \ref{lem:BCP}.
Then we have
\begin{align}
  &
  E_{\sfA\sfB\sfcc}\lrB{\Error(\sfA,\sfB,\sfcc)}
  \notag
  \\*
  &=
  E_{\sfA\sfB}\lrB{
    \sum_{\substack{
	\cc,\mm:\\
	\mu_{X^n}(\C_{\sfA\sfB}(\cc,\mm))=0
    }}
    \frac1
    {|\im\A||\im\B|}
    +
    \sum_{\substack{
	\cc,\mm,\xx,\yy:\\
	\mu_{X^n}(\C_{\sfA\sfB}(\cc,\mm))>0\\
	\xx\in\C_{\sfA\sfB}(\cc,\mm)\\
	\xx_{\sfA}(\cc|\yy)\neq\xx
    }}
    \frac{\mu_{XY}(\xx,\yy)}
    {|\im\A||\im\B|\mu_X(\C_{\sfA\sfB}(\cc,\mm))}
  }
  \notag
  \\
  &=
  E_{\sfA\sfB}\left[
    \sum_{\substack{
	\cc,\mm:\\
	\mu_{X^n}(\C_{\sfA\sfB}(\cc,\mm))=0
    }}
    \frac1
    {|\im\A||\im\B|}
    \vphantom{
      \sum_{\cc,\mm,\xx\in\C_{\sfA\sfB}(\cc,\mm),\yy}
      \sum_{\substack{
	  \cc,\mm,\xx,\yy:\\
	  \mu_{X^n}(\C_{\sfA\sfB}(\cc,\mm))>0\\
	  \xx\in\C_{\sfA\sfB}(\cc,\mm)\\
	  \xx_{\sfA}(\cc|\yy)\neq\xx
      }}
    }
  \right.
  \notag
  \\*
  &\qquad\qquad
  +
  \left.
    \sum_{\cc,\mm,\xx\in\C_{\sfA\sfB}(\cc,\mm),\yy}
    \sum_{\substack{
	\cc,\mm,\xx,\yy:\\
	\mu_{X^n}(\C_{\sfA\sfB}(\cc,\mm))>0\\
	\xx\in\C_{\sfA\sfB}(\cc,\mm)\\
	\xx_{\sfA}(\cc|\yy)\neq\xx
    }}
    \mu_{XY}(\xx,\yy)
    \vphantom{
      \sum_{\cc,\mm,\xx\in\C_{\sfA\sfB}(\cc,\mm),\yy}
      \mu_{XY}(\xx,\yy)
    }
    \lrB{
      1+
      \frac1
      {|\im\A||\im\B|\mu_X(\C_{\sfA\sfB}(\cc,\mm))}
      -
      1
    }
  \right]
  \notag
  \\
  &\leq
  E_{\sfA}\lrB{
    \sum_{\xx,\yy}
    \mu_{XY}(\xx,\yy)
    \chi(\xx_{\sfA}(\sfA\xx|\yy)\neq\xx)
  }
  +
  E_{\sfA\sfB}\lrB{
    \sum_{\cc,\mm}
    \lrbar{
      \mu_X(\C_{AB}(\cc,\mm))
      -
      \frac1{|\im\A||\im\B|}
    }
  }
  \notag
  \\
  \begin{split}
    &\leq
    2^{-n[r-\oH(\XX|\YY)-\e]}\alphaA{}
    +\betaA{}
    +\mu_{XY}([\oT_{X|Y}]^c)
    \\*
    &\quad
    +
    \sqrt{
      \alphaAB-1
      +
      [\betaAB+1]2^{-n[\uH(\XX)-r-R-\e]}
    }
    +2\mu_X([\uT_X]^c),
  \end{split}
  \label{eq:proof-channel-error}
\end{align}
where $\sfcc$ is a random variable corresponding to the uniform
distribution on $\im\A$, the first inequality comes from the fact that
\begin{align}
  &
  \sum_{\substack{
      \cc,\mm,\xx,\yy:\\
      \mu_X(\C_{AB}(\cc,\mm))>0
      \\
      \xx\in\C_{AB}(\cc,\mm)
  }}
  \mu_{XY}(\xx,\yy)
  \lrB{
    \frac1
    {|\im\A||\im\B|\mu_X(\C_{AB}(\cc,\mm))}
    -
    1
  }
  \notag
  \\
  &\leq
  \sum_{\substack{
      \cc,\mm:\\
      \mu_X(\C_{AB}(\cc,\mm))>0
  }}
  \lrbar{
    \frac1
    {|\im\A||\im\B|\mu_X(\C_{AB}(\cc,\mm))}
    -
    1
  }
  \mu_X(\C_{AB}(\cc,\mm))
  \notag
  \\
  &=
  \sum_{\substack{
      \cc,\mm:\\
      \mu_X(\C_{AB}(\cc,\mm))>0
  }}
  \lrbar{
    \mu_X(\C_{AB}(\cc,\mm))
    -
    \frac1
    {|\im\A||\im\B|}
  }
  \notag
  \\
  &=
  \sum_{\cc,\mm}
  \lrbar{
    \mu_X(\C_{AB}(\cc,\mm))
    -
    \frac1
    {|\im\A||\im\B|}
  }
  -
  \sum_{\substack{
      \cc,\mm:\\
      \mu_X(\C_{AB}(\cc,\mm))=0
  }}
  \frac1
  {|\im\A||\im\B|},
\end{align}
and the second inequality comes from (\ref{eq:proof-channel-error1}),
(\ref{eq:proof-channel-error2}).
From  (\ref{eq:proof-channel-r-e}), (\ref{eq:proof-channel-rR-e}),
(\ref{eq:proof-channel-error}) and the fact that
$\alphaA{}\to1$, $\betaA{}\to0$, $\alphaAB{}\to1$, $\betaAB{}\to0$,
$\mu_X([\uT_X]^c)\to0$, $\mu_{XY}([\uT_{XY}]^c)\to0$ as $n\to\infty$,
we have the fact that there are functions $A\in\A$, $B\in\B$, and a vector
$\cc\in\im\A$ satisfying (\ref{eq:channel-error}).
\hfill\qed

\subsection{Proof of Corollary~\ref{thm:channel-uniform}}
\label{sec:proof-channel-uniform}

Inequality (\ref{eq:channel-uniform-R}) is shown as
\begin{align*}
  R
  &\equiv
  \frac1n\log|\M_n|
  \notag
  \\
  &=
  \frac1n\log\frac{|\X^n|}{|\im A|}
  \notag
  \\
  &\geq
  \log|\X|-r,
\end{align*}
where the inequality comes from the
definition of $r$ and the fact that $\im A\subset\im\A$.

Since $\mu_{X^n}$ is uniform
and for given $\cc\in\im A$ and $\mm\in\M_n$
there is a unique $\xx\in\C_{AB}(\cc,\mm)$,
we have the fact that
\[
  \frac 1{|\C_A(\cc)|}
  =\frac{\mu_{X^n}(\xx)}
  {|\M_n|\mu_{X^n}(\C_{AB}(\cc,\mm))}
\]
for all $\mm$.
Then we have
\begin{align}
  \begin{split}
    E_{\sfA\sfcc}\lrB{\Error(\sfA,\sfcc)}
    &\leq
    2^{-n[r-\oH(\XX|\YY)-\e]}\alphaA{}
    +\betaA{}
    +\mu_{XY}([\oT_{X|Y}]^c)
    \\*
    &\quad
    +
    \sqrt{
      \alphaAB-1
      +
      [\betaAB+1]2^{-n[\uH(\XX)-r-R-\e]}
    }
    +2\mu_X([\uT_X]^c),
  \end{split}
  \label{eq:proof-channel-uniform-error}
\end{align}
from (\ref{eq:proof-channel-error}).
From (\ref{eq:channel-uniform-r}), (\ref{eq:proof-channel-uniform-error}),
and the fact that $\alphaA{}\to1$, $\betaA{}\to0$,
$\mu_{XY}([\oT_{X|Y}]^c)\to0$, $\mu_X([\uT_X]^c)\to0$ as $n\to\infty$,
we have the fact that for any $\delta>0$ and sufficiently large $n$ there
are functions $A\in\A$, and a vector $\cc\in\im\A$ satisfying
(\ref{eq:channel-uniform-error}) for all $\delta>0$ and sufficiently
large $n$.

Now, we prove (\ref{eq:capacity-additive}) following the proof presented
in~\cite{VH94}\cite[Example 3.2.1]{HAN}.
Assume that $\mu_{Y^n|X^n}$ is a channel with additive noise
$\ZZ=\{Y^n-X^n\}_{n=1}^{\infty}$.
Since the channel $\mu_{Y^n|X^n}$ is weakly symmetric (see \cite[p.190]{CT}),
then the reverse channel $\mu_{X^n|Y^n}$ is also weakly symmetric
when the channel input distribution $\mu_{X^n}$ is uniform.
This implies that $\oH(\XX|\YY)$ does not depend on $\YY$ and
\[
  \oH(\XX|\YY)=\oH(\XX|\zero)=\oH(-\ZZ)=\oH(\ZZ).
\]
We have
\begin{align}
  \uI(\XX;\YY)
  &\leq\oH(\XX)-\oH(\XX|\YY)
  \notag
  \\
  &\leq \log|\X|-\oH(\XX|\YY)
  \notag
  \\
  &\leq \log|\X|-\oH(\ZZ).
\end{align}
This implies that $\log|\X|-\oH(\ZZ)\geq C(\WW)$.
On the other hand, the supremum on the right hand side of
(\ref{eq:capacity}) is achieved by assuming that $\mu_{X^n}$ is the uniform
distribution on $\X^n$.
This implies that $\log|\X|-\oH(\ZZ)$ is the capacity of this channel.
\hfill\qed

\subsection{Proof of Theorem~\ref{thm:lossy}}
\label{sec:proof-lossy}
We omit the dependence on $n$ of $X$ and $Y$ when they appear
in the subscript of $\mu$.

From (\ref{eq:lossy-r}) and (\ref{eq:lossy-rR}),
we have the fact that there is $\e>0$ satisfying
\begin{align}
  r
  &<
  \uH(\XX|\YY)-\e
  \label{eq:proof-lossy-r}
  \\
  r+R
  &>
  \oH(\XX)+\e.
  \label{eq:proof-lossy-rR}
\end{align}

Let $\oT_X\subset\X^n$ and $\uT_{X|Y}\subset\X^n\times\Y^n$
be defined as
\begin{align*}
  \oT_X
  &\equiv\lrb{\xx: \frac1n\log\frac1{\mu_X(\xx)}\leq \oH(\XX)+\e }
  \\
  \uT_{X|Y}
  &\equiv
  \lrb{
    (\xx,\yy):
    \frac 1n\log\frac 1{\mu_{X|Y}(\xx|\yy)} \geq \oH(\XX|\YY)-\e
  }.
\end{align*}
Assume that $\xx\in\oT_X$ and $\xx_{AB}(A\xx,B\xx)\neq \xx$.
Then we have the fact that there is $\xx'\in\C_{AB}(A\xx,B\xx)$ such that
$\xx'\neq \xx$ and
\[
  \mu_X(\xx')\geq
  \mu_X(\xx)\geq 2^{-n[\oH(\XX)+\e]}.
\]
This implies that $\lrB{\oT_X\setminus\{\xx\}}\cap\C_{AB}(A\xx,B\xx)\neq\emptyset$.
Then we have
\begin{align}
  E_{\sfA\sfB}\lrB{
    \chi(\xx_{\sfA}(\sfA\xx,\sfB\xx)\neq\xx)
  }
  &\leq
  p_{\sfA\sfB}\lrsb{\lrb{(A,B):
      [\oT_X\setminus\{\xx\}]\cap\C_{AB}(A\xx,B\xx)\neq\emptyset
  }}
  \notag
  \\*
  &\leq
  \frac{|\oT_X|\alphaAB{}}
  {|\im\A|}
  +\betaAB{}
  \notag
  \\
  &\leq
  2^{-n[r-\oH(\XX)-\e]}\alphaAB{}
  +\betaAB{},
  \label{eq:proof-lossy-ave1-sub}
\end{align}
where $\chi(\cdot)$ is defined by (\ref{eq:chi}),
the second inequality comes from Lemma~\ref{lem:CRP},
and the last inequality comes from the fact that
$|\oT_X|\leq 2^{n[\oH(\XX)+\e]}$.
We have the fact that
\begin{align}
  &
  E_{\sfA\sfB}\lrB{
    \sum_{\xx}
    \mu_X(\xx)
    \chi(\xx_{\sfA\sfB}(\sfA\xx,\sfB\xx)\neq\xx)
  }
  \notag
  \\*
  &=
  \sum_{\xx\in\oT_X}
  \mu_X(\xx)
  E_{\sfA\sfB}\lrB{
    \chi(\xx_{\sfA\sfB}(\sfA\xx,\sfB\xx)\neq\xx)
  }
  +
  \sum_{\xx\notin\oT_X}
  \mu_X(\xx)
  E_{\sfA\sfB}\lrB{
    \chi(\xx_{\sfA\sfB}(\sfA\xx,\sfB\xx)\neq\xx)
  }
  \notag
  \\
  &\leq
  2^{-n[r+R-\oH(\XX)-\e]}\alphaAB{}
  +\betaAB{}
  +\mu_X([\oT_X]^c),
  \label{eq:proof-lossy-ave1}
\end{align}
where the last inequality comes from (\ref{eq:proof-lossy-ave1-sub}).
We also have the fact that
\begin{align}
  &
  E_{\sfA}\lrB{
    \sum_{\cc,\yy}
    \mu_Y(\yy)
    \lrbar{
      \mu_{X|Y}(\C_{\sfA}(\cc)|\yy)
      -
      \frac1
      {|\im\A|}
    }
  }
  \notag
  \\*
  &\leq
  E_{\sfA}\left[
    \sum_{\cc,\yy}
    \mu_{X|Y}(\uT_{X|Y}(\yy)|\yy)
    \mu_Y(\yy)
    \lrbar{
      \frac{\mu_{X|Y}(\C_{\sfA}(\cc)\cap\uT_{X|Y}(\yy)|\yy)}
      {\mu_{X|Y}(\uT_{X|Y}(\yy)|\yy)}
      -
      \frac1
      {|\im\A|}
    }
  \right]
  \notag 
  \\*
  &\quad
  +
  E_{\sfA}\lrB{
    \sum_{\cc,\yy}
    \mu_{X|Y}(\C_{\sfA}(\cc)\cap[\uT_{X|Y}(\yy)]^c|\yy)
    \mu_Y(\yy)
  }
  +
  E_{\sfA}\lrB{
    \sum_{\cc,\yy}
    \frac{\mu_{X|Y}([\uT_{X|Y}(\yy)]^c|\yy)\mu_Y(\yy)}
    {|\im\A|}
  }
  \notag
  \\
  &=
  \sum_{\yy}
  \mu_{X|Y}(\uT_{X|Y}(\yy)|\yy)\mu_Y(\yy)
  E_{\sfA}\lrB{
    \sum_{\cc}
    \lrbar{
      \frac{\mu_{X|Y}(\C_{\sfA}(\cc)\cap\uT_{X|Y}(\yy)|\yy)}
      {\mu_{X|Y}(\uT_{X|Y}(\yy)|\yy)}
      -
      \frac1
      {|\im\A|}
    }
  }
  \notag
  \\*
  &\quad
  +
  2\mu_{XY}([\uT_{X|Y}]^c)
  \notag
  \\
  &\leq
  \sum_{\yy}
  \mu_{X|Y}(\uT_{X|Y}(\yy)|\yy)\mu_Y(\yy)
  \sqrt{
    \frac{\alphaAB{}-1+[\betaAB{}+1]|\im\A|\displaystyle\max_{\xx\in\uT_X}\mu_X(\xx)}
    {\mu_{X|Y}(\uT_{X|Y}(\yy)|\yy)}
  }
  +
  2\mu_{X}([\uT_X]^c)
  \notag
  \\*
  &\leq
  \sqrt{\alphaA{}{}-1+[\betaA{}{}+1]2^{-n[\uH(\XX|\YY)-r-\e]}}
  +
  2\mu_{X|Y}([\uT_{X|Y}]^c),
  \label{eq:proof-lossy-ave2}
\end{align}
where the second inequality comes from Lemma \ref{lem:BCP}.
Then we have
\begin{align}
  &
  E_{\sfA\sfB\sfcc}\lrB{\Error(\sfA,\sfB,\sfcc,D)}
  \notag
  \\*
  &\leq
  E_{\sfA\sfB\sfcc}\left[
    \sum_{\substack{
	\yy:\\
	\mu_{X|Y}(\C_{\sfA}(\sfcc)|\yy)=0
    }}
    \mu_Y(\yy)
    +
    \sum_{\substack{
	\xx,\yy:\\
	\xx\in\C_{\sfA}(\sfcc)\\
	\mu_{X|Y}(\C_{\sfA}(\sfcc)|\yy)>0\\
	d_n(\xx,\yy)>D\ \text{or}\ \xx_{\sfA\sfB}(\sfcc,\sfB\xx)\neq\xx
    }}
    \frac{\mu_{X|Y}(\xx|\yy)\mu_Y(\yy)}
    {\mu_{X|Y}(\C_{\sfA}(\sfcc)|\yy)}
  \right]
  \notag
  \\
  &=
  E_{\sfA\sfB}\left[
    \sum_{\substack{
	\cc,\yy:\\
	\mu_{X|Y}(\C_{\sfA}(\cc)|\yy)=0
    }}
    \frac{\mu_Y(\yy)}{|\im\A|}
    \vphantom{
      \sum_{\substack{
	  \cc,\xx,\yy:\\
	  \xx\in\C_{\sfA}(\cc)\\
	  \mu_{X|Y}(\C_{\sfA}(\cc)|\yy)>0\\
	  d_n(\xx,\yy)>D\ \text{or}\ \xx_{\sfA\sfB}(\sfA\xx,\sfB\xx)\neq\xx
      }}
    }
  \right.
  \notag
  \\*
  &\qquad\qquad
  \left.
    +
    \sum_{\substack{
	\cc,\xx,\yy:\\
	\xx\in\C_{\sfA}(\cc)\\
	\mu_{X|Y}(\C_{\sfA}(\cc)|\yy)>0\\
	d_n(\xx,\yy)>D\ \text{or}\ \xx_{\sfA\sfB}(\sfA\xx,\sfB\xx)\neq\xx
    }}
    \mu_{XY}(\xx,\yy)
    \lrB{
      1+
      \frac1
      {|\im\A|\mu_{X|Y}(\C_{\sfA}(\cc)|\yy)}
      -
      1
    }
  \right]
  \notag
  \\
  &\leq
  P(d_n(X^n,Y^n)>D)
  +
  E_{\sfA\sfB}\lrB{
    \sum_{\xx}
    \mu_X(\xx)
    \chi(\xx_{\sfA\sfB}(\sfA\xx,\sfB\xx)\neq\xx)
  }
  \notag
  \\*
  &\quad
  +
  E_{\sfA}\lrB{
    \sum_{\cc,\yy}
    \mu_Y(\yy)
    \lrbar{
      \mu_{X|Y}(\C_{\sfA}(\cc)|\yy)
      -
      \frac1
      {|\im\A|}
    }
  }
  \notag
  \\
  \begin{split}
    &\leq
    P(d_n(X^n,Y^n)>D)
    +
    2^{-n[r+R-\oH(\XX)-\e]}\alphaAB{}
    +\betaAB{}
    +\mu_X([\oT_X]^c)
    \\*
    &\quad
    +
    \sqrt{
      \alphaA{}-1
      +
      [\betaA{}+1]2^{-n[\uH(\XX|\YY)-r-\e]}
    }
    +2\mu_{XY}([\uT_{X|Y}]^c),
  \end{split}
  \label{eq:proof-lossy-error}
\end{align}
where $\sfcc$ is a random variable corresponding to the uniform
distribution on $\im\A$, the second inequality comes from the fact that
\begin{align}
  &
  \sum_{\substack{
      \cc,\xx,\yy:\\
      \xx\in\C_A(\cc)\\
      \mu_{X|Y}(\C_A(\cc)|\yy)>0\\
  }}
  \mu_{XY}(\xx,\yy)
  \lrB{
    \frac1
    {|\im\A|\mu_{X|Y}(\C_A(\cc)|\yy)}
    -
    1
  }
  \notag
  \\*
  &\leq
  \sum_{\substack{
      \cc,\yy:\\
      \mu_{X|Y}(\C_A(\cc)|\yy)>0
  }}
  \mu_{X|Y}(\C_A(\cc)|\yy)\mu_Y(\yy)
  \lrbar{
    \frac1
    {|\im\A|\mu_{X|Y}(\C_A(\cc)|\yy)}
    -
    1
  }
  \notag
  \\
  &=
  \sum_{\substack{
      \cc,\yy:\\
      \mu_{X|Y}(\C_A(\cc)|\yy)>0
  }}
  \mu_Y(\yy)
  \lrbar{
    \mu_X(\C_A(\cc)|\yy)
    -
    \frac1
    {|\im\A|}
  }
  \notag
  \\
  &=
  \sum_{\cc,\yy}
  \mu_Y(\yy)
  \lrbar{
    \mu_X(\C_A(\cc)|\yy)
    -
    \frac1
    {|\im\A|}
  }
  -
  \sum_{\substack{
      \cc,\yy:\\
      \mu_{X|Y}(\C_A(\cc)|\yy)=0
  }}
  \frac{\mu_Y(\yy)}
  {|\im\A|},
\end{align}
and the third inequality comes from (\ref{eq:proof-lossy-ave1}),
(\ref{eq:proof-lossy-ave2}).
From (\ref{eq:proof-lossy-r}), (\ref{eq:proof-lossy-rR}),
(\ref{eq:proof-lossy-error}) and the fact that $\alphaA{}\to1$,
$\betaA{}\to0$, $\alphaAB{}\to1$, $\betaAB{}\to0$,
$\mu_X([\uT_X]^c)\to0$, $\mu_{XY}([\uT_{X|Y}]^c)\to0$ as $n\to\infty$,
we have the fact that there are functions $A\in\A$, $B\in\B$, and a vector
$\cc\in\im\A$ satisfying (\ref{eq:lossy-error-bound}).
\hfill\qed

\subsection{Proof of Corollary~\ref{thm:lossy-uniform}}
\label{sec:proof-lossy-uniform}

Since $\xx'_{AB}(\cc,B\xx)=\xx$ is satisfied for all $\xx$,
we can substitute
\[
  \chi(\xx_{AB}(\cc,B\xx)\neq\xx)=0
\]
in the derivation of (\ref{eq:proof-lossy-error})
and obtain
\begin{align}
  &
  E_{\sfA\sfcc}\lrB{\Error(\sfA,\sfcc,D)}
  \notag
  \\*
  &\leq
  P\lrsb{d(X^n,Y^n)>D}
  +
  \sqrt{
    \alphaA{}-1
    +
    [\betaA{}+1]2^{-n[\uH(\XX|\YY)-r-\e]}
  }
  +2\mu_{XY}([\uT_{X|Y}]^c).
  \label{eq:proof-lossy-uniform-error}
\end{align}
On the other hand, from Lemma~\ref{lem:BCP}, we have
\begin{align}
  E_{\sfA\sfcc}\lrB{
    \lrbar{
      \frac{|\im\A||\C_{\sfA}(\sfcc)|}{|\X^n|}-1
    }
  }
  &
  =
  E_{\sfA}\lrB{
    \sum_{\cc}\lrbar{
      \frac{|\C_{\sfA}(\cc)|}{|\X^n|}-\frac1{|\im\A|}
    }
  }
  \notag
  \\*
  &\leq
  \sqrt{
    \alphaA{}-1
    +
    \frac{[\betaA{}+1]|\im\A|}{|\X^n|}
  }
  \notag
  \\
  &=
  \sqrt{
    \alphaA{}-1
    +
    [\betaA{}+1]2^{-n[\log|\X|-r]}
  }.
\end{align}
By using the Markov inequality,
(\ref{eq:lossy-uniform-r}), (\ref{eq:proof-lossy-uniform-error}),
and the fact that $\alphaA{}\to1$, $\betaA{}\to0$,
$\mu_{XY}([\uT_{X|Y}]^c)\to0$ as $n\to\infty$,
we have the fact that for any $\delta>0$ and sufficiently large $n$ there
are functions $A\in\A$, and a vector $\cc\in\im\A$ satisfying
(\ref{eq:lossy-uniform-d}) and
\begin{equation}
  \lrbar{
    \frac{|\im\A||\C_A(\cc)|}{|\X^n|}-1
  }
  <
  1
  \label{eq:proof-lossy-uniform-rate}
\end{equation}
for sufficiently large $n$.
Then we have the fact that $\cc\in\im A\subset\im\A$
because the left hand side of (\ref{eq:proof-lossy-uniform-rate})
is equal to $1$ when $\cc\in\im\A\setminus\im A$.
From (\ref{eq:proof-lossy-uniform-rate}) and the fact that $A$ is a
linear function, we have
\begin{align}
  \frac{|\im\A||\C_A(\zero)|}{|\X^n|}
  &=
  \frac{|\im\A||\C_A(\cc)|}{|\X^n|}
  <
  2
\end{align}
and
\begin{align}
  R
  &=
  \frac 1n\log|\C_A(\zero)|
  \notag
  \\
  &\leq
  \frac 1n\log\frac{2|\X^n|}{|\im\A|}
  \notag
  \\
  &\leq \log|\X|-r+\delta,
\end{align}
for all $\delta>0$ and sufficiently large $n$.
\hfill\qed

\subsection{Proof of Theorem~\ref{thm:sum-product}}
\label{sec:proof-sum-product}.
Let $g'_0$ and $g'_k:\X\to\Z_k$ be defined as
\begin{align}
  g'_0
  &\equiv
  \sum_{\xx}
  \prod_{j=1}^n
  \mu_{X_j}(x_j)
  \prod_{i=1}^{l}\chi(\ba_i(\xx_{\cS_i})=c_i)
  \label{eq:sum-product-fp0}
  \\
  g'_k(x_k)
  &\equiv
  \mu_{X_k}(x_k)
  \sum_{x_{k+1}^n}
  \prod_{j=k+1}^n
  \mu_{X_j}(x_j)
  \prod_{i=1}^{l}\chi(\ba_i(\xx_{\cS_i})=c_i)
\end{align}
for a given $\cc\equiv(c_1,\ldots,c_l)\in\Z^l$.
Then we have
\begin{align}
  p_{\tX_k|\tX_1^{k-1}}(x_k|x_1^{k-1})
  &=
  \frac{
    \sum_{x_{k+1}^n}
    \prod_{j=k}^n
    p_{X_j}(x_j)
    \prod_{i=1}^{l}\chi(\ba_i(\xx_{\cS_i})=c_i)
  }{
    \sum_{x_k^n}
    \prod_{j=k}^n
    p_{X_j}(x_j)
    \prod_{i=1}^{l}\chi(\ba_i(\xx_{\cS_i})=c_i)
  }
  \notag
  \\
  &=
  \frac{g'_k(x_k)}{\sum_{x_k}g'_k(x_k)}.
  \label{eq:sum-product-ptX}
\end{align}

If the algorithm terminates with $k=n$ at {\sf Step 4},
we have
\begin{equation}
  g'_n(x_n)=p_{X_n}(x_n).
  \label{eq:sum-product-fpn}
\end{equation}
On the other hand, if  the algorithm terminates with $k=k'$
at {\sf Step 5}, we have
\begin{align}
  g'_{k'}(x_{k'})
  &=
  p_{X_{k'}}(x_{k'})
  \sum_{x_{k'+1}^n}
  \prod_{j=k'+1}^n
  p_{X_j}(x_j)
  \prod_{i=1}^{l}\chi(\ba_i(\xx_{\cS_i})=c_i)
  \notag
  \\
  &=
  \prod_{j=k'}^n
  p_{X_j}(x_j),
  \label{eq:sum-product-fpkp}
\end{align}
where the second equality comes from the fact that
for a given $x_1^{k'}$
there is a unique $x_{k'+1}^n$
such that $x_1^n\in\C_A(\xx)$.
Since (\ref{eq:sum-product-fpn}) is a special case of
(\ref{eq:sum-product-fpkp}) with $k'=n$,
we assume that the algorithm terminates at $k=k'$ in the following.

Since
\begin{align}
  \sum_{x_1}g'_1(x_1)
  &=
  \sum_{x_1}
  \mu_{X_1}(x_1)
  \sum_{x_2^n}
  \prod_{j=2}^n
  \mu_{X_j}(x_j)
  \prod_{i=1}^{l}\chi(\ba_i(\xx_{\cS_i})=c_i)
  \notag
  \\
  &=
  g'_0
  \label{eq:sum-product-sumfp1}
\end{align}
and
\begin{align}
  \sum_{x_k}g'_k(x_k)
  &=
  \sum_{x_k}p_{X_k}(x_k)
  \sum_{x_{k+1}^n}
  \prod_{j=k+1}^n
  \mu_{X_j}(x_j)
  \prod_{i=1}^{l}\chi(\ba_i(\xx_{\cS_i})=c_i)
  \notag
  \\
  &=
  \sum_{x_k^n}
  \prod_{j=k}^n
  \mu_{X_j}(x_j)
  \prod_{i=1}^{l}\chi(\ba_i(\xx_{\cS_i})=c_i)
  \notag
  \\
  &=
  \frac{g'_{k-1}(x_{k-1})}{\mu_{X_{k-1}}(x_{k-1})}
  \label{eq:sum-product-sumfpk}
\end{align}
for $k\geq 2$, we have the fact that
(\ref{eq:sum-product-tX})
is rephrased as
\begin{align}
  \mu_{\tX^n}(\xx)
  &=
  \frac{\prod_{j=1}^n\mu_{X_i}(x_i)\prod_{i=1}^l\chi(\ba_i(\xx_{\cS_1})=c_i)}
  {\sum_{\xx}\prod_{j=1}^n\mu_{X_i}(x_i)\prod_{i=1}^l\chi(\ba_i(\xx_{\cS_1})=c_i)}
  \notag
  \\
  &=
  \frac{g'_{k'}(x_{k'})}{g'_0}\prod_{j=1}^{k'-1}\mu_{X_j}(x_j)
  \notag
  \\
  &=
  \frac{g'_1(x_1)}{g'_0}
  \prod_{k=2}^{k'}\frac{\mu_{X_{k-1}}(x_k)g'_k(x_k)}{g'_{k-1}(x_{k-1})}
  \notag
  \\
  &=
  \prod_{k=1}^{k'}\frac{\mu_{X_{k-1}}(x_k)g'_k(x_k)}{g'_{k-1}(x_{k-1})},
  \notag
  \\
  &=
  \prod_{k=1}^{k'}\frac{g'_k(x_k)}{\sum_{x_k}g'_k(x_k)}
  \notag
  \\
  &=
  \prod_{k=1}^{k'}p_{\tX_k|\tX_1^{k=1}}(x_k|x_1^{k-1}),
\end{align}
where the first equality comes from (\ref{eq:sum-product-tX}),
the second equality comes from (\ref{eq:sum-product-fp0}),
(\ref{eq:sum-product-fpkp}),
we denote $g'_0(x_0)\equiv g'_0$ in the fourth equality,
the fifth equality comes from (\ref{eq:sum-product-sumfp1}),
(\ref{eq:sum-product-sumfpk}),
and the last equality comes from (\ref{eq:sum-product-ptX}).

Since the algorithm generates a sequence $\xx\equiv x_1^n$
subject to $\prod_{k=1}^{k'}p_{\tX_k|\tX_1^{k=1}}(x_k|x_1^{k-1})$,
the proposed algorithm generates $\xx$
subject to the probability distribution given by
(\ref{eq:sum-product-tX}).\hfill\qed

\appendix

We prove the lemmas used in the proofs of the theorems.
Some proofs are presented for the completeness of this paper.

\subsection{Lemma Analogous to Fano Inequality}
\label{sec:fano}
We prove the following lemma which is analogous to the Fano inequality.
It should be noted that a stronger version of this lemma has been proved
in \cite[Lemma 4]{K08}.
\begin{lem}
\label{lem:fano}
Let $(\UU,\VV)\equiv\{(U^n,V^n)\}_{n=1}^{\infty}$ be a pair consisting
of two sequences of random variables.
If there is $\{\psi_n\}_{n=1}^{\infty}$ such that
\begin{equation}
  \limn P(\psi_n(V^n)\neq U^n)=0,
  \label{eq:fano-error}
\end{equation}
then
\begin{equation}
  \oH(\UU|\VV)=0.
  \label{eq:fano-H}
\end{equation}
\end{lem}
\begin{proof}
For $\gamma>0$, let
\begin{align*}
  \G
  &\equiv\lrb{
    (\uu,\vv): \frac 1n\log\frac1{\mu_{U^n|V^n}(\uu|\vv)}\geq \gamma
  }
  \\
  \cS
  &\equiv\lrb{
    (\uu,\vv): \psi_n(\vv)=\uu
  }.
\end{align*}
Then we have
\begin{align}
  \mu_{U^nV^n}(\G)
  &=
  \mu_{U^nV^n}(\G\cap\cS^c)+\mu_{U^nV^n}(\G\cap\cS)
  \notag
  \\
  &=
  \mu_{U^nV^n}(\G\cap\cS^c)
  +\sum_{(\uu,\vv)\in\G\cap\cS}\mu_{U^nV^n}(\uu,\vv)
  \notag
  \\
  &=
  \mu_{U^nV^n}(\G\cap\cS^c)
  +
  \sum_{\vv}\mu_{V^n}(\vv)
  \sum_{\substack{
      \uu:\\
      \psi_n(\vv)=\uu\\
      (\uu,\vv)\in\G
  }}
  \mu_{U^n|V^n}(\uu|\vv)
  \notag
  \\
  &\leq
  \mu_{U^nV^n}(\G\cap\cS^c)
  +
  \sum_{\vv}\mu_{V^n}(\vv)
  \sum_{\uu: \psi_n(\vv)=\uu}
  2^{-n\gamma}
  \notag
  \\
  &\leq
  P(\psi_n(V^n)\neq U^n)+2^{-n\gamma},
\end{align}
where the first inequality comes from the definition of $\G$
and the last inequality comes from the fact that for all $\vv$ there is
a unique $\uu$ satisfying $\psi_n(\vv)=\uu$.
From this inequality and (\ref{eq:fano-error}), we have
\[
  \limn\Prob\lrsb{
    \frac 1n\log\frac1{\mu_{U^n|V^n}(U^n|V^n)}\geq \gamma
  }=0.
\]
Then we have
\[
  0\leq\oH(\UU|\VV)\leq \gamma
\]
from the definition of $\oH(\UU|\VV)$.
We have (\ref{eq:fano-H}) by letting $\gamma\to0$.
\end{proof}

\subsection{Proof of (\ref{eq:whash})}
\label{sec:proof-whash}
If an ensemble satisfies (\ref{eq:hash}), then we have
\begin{align}
  \sum_{\substack{
      \uu\in\T
      \\
      \uu'\in\T'
  }}
  \pA\lrsb{\lrb{A: A\uu = A\uu'}}
  &=
  \sum_{\uu\in\T\cap\T'}
  \pA\lrsb{\lrb{A: A\uu = A\uu'}}
  \notag
  \\*
  &\quad
  +
  \sum_{\uu\in\T}
  \sum_{\substack{
      \uu'\in\T'\setminus\{\uu\}:
      \\
      \pA(\{A: A\uu = A\uu'\})\leq\frac{\alphaA}{|\im\A|}
  }}
  \pA\lrsb{\lrb{A: A\uu = A\uu'}}
  \notag
  \\*
  &\quad
  +
  \sum_{\uu\in\T}
  \sum_{\substack{
      \uu'\in\T'\setminus\{\uu\}:
      \\
      \pA(\{A: A\uu = A\uu'\})>\frac{\alphaA}{|\im\A|}
  }}
  \pA\lrsb{\lrb{A: A\uu = A\uu'}}
  \notag
  \\
  &\leq
  |\T\cap\T'|
  +
  \sum_{\uu\in\T}
  \sum_{\substack{
      \uu'\in\T'\setminus\{\uu\}:
      \\
      \pA(\{A: A\uu = A\uu'\})\leq\frac{\alphaA}{|\im\A|}
  }}
  \frac{\alphaA}{|\im\A|}
  +
  \sum_{\uu\in\T}
  \betaA
  \notag
  \\
  &\leq
  |\T\cap\T'|
  +
  \frac{|\T||\T'|\alphaA}{|\im\A|}
  +
  |\T|\betaA
  \notag
  \\
  &\leq
  |\T\cap\T'|
  +
  \frac{|\T||\T'|\alphaA}{|\im\A|}
  +
  \min\{|\T|,|\T'|\}\betaA
\end{align} 
for any $\T$ and $\T'$ satisfying $|\T|\leq |\T'|$.
\hfill\QED

\subsection{Proof of Lemma \ref{lem:hash-AB}}
\label{sec:proof-hash-AB}
Let
\begin{align*}
  p_{\sfA,\uu,\uu'}
  &\equiv
  \pA(\{A: A\uu=A\uu'\})
  \\
  p_{\sfB,\uu,\uu'}
  &\equiv
  \pB(\{B: B\uu=B\uu'\}).
  \\
  p_{\sfA\sfB,\uu,\uu'}
  &\equiv
  \pAB(\{(A,B): (A,B)\uu=(A,B)\uu'\}).
\end{align*}
Then we have
\begin{align}
  &\sum_{\substack{
      \uu'\in\U^n\setminus\{\uu\}:
      \\
      p_{\sfA\sfB,\uu,\uu'}>
      \frac{\alpha_{\sfA\sfB}}{\lrbar{\im[\A\times\B]}}
  }}
  p_{\sfA\sfB,\uu,\uu'}
  \notag
  \\*
  &\leq
  \sum_{\substack{
      \uu'\in\U^n\setminus\{\uu\}:
      \\
      p_{\sfA,\uu,\uu'}p_{\sfB,\uu,\uu'}
      >\frac{\alphaA\alphaB}{|\im\A||\im\B|}
  }}
  p_{\sfA,\uu,\uu'}p_{\sfB,\uu,\uu'}
  \notag
  \\
  &=
  \sum_{\substack{
      \uu'\in\U^n\setminus\{\uu\}:
      \\
      p_{\sfA,\uu,\uu'}p_{\sfB,\uu,\uu'}
      >\frac{\alphaA\alphaB}{|\im\A||\im\B|}
      \\
      p_{\sfA,\uu,\uu'}>\frac{\alphaA}{|\im\A|}
  }}
  p_{\sfA,\uu,\uu'}p_{\sfB,\uu,\uu'}
  +
  \sum_{\substack{
      \uu'\in\U^n\setminus\{\uu\}:
      \\
      p_{\sfA,\uu,\uu'}p_{\sfB,\uu,\uu'}
      >\frac{\alphaA\alphaB}{|\im\A||\im\B|}
      \\
      p_{\sfA,\uu,\uu'}\leq\frac{\alphaA}{|\im\A|}
  }}
  p_{\sfA,\uu,\uu'}p_{\sfB,\uu,\uu'}
  \notag
  \\
  &\leq
  \sum_{\substack{
      \uu'\in\U^n\setminus\{\uu\}:
      \\
      p_{\sfA,\uu,\uu'}>\frac{\alphaA}{|\im\A|}
  }}
  p_{\sfA,\uu,\uu'}p_{\sfB,\uu,\uu'}
  +
  \sum_{\substack{
      \uu'\in\U^n\setminus\{\uu\}:
      \\
      p_{\sfB,\uu,\uu'}>\frac{\alphaB}{|\im\B|}
  }}
  p_{\sfA,\uu,\uu'}p_{\sfB,\uu,\uu'}
  \notag
  \\
  &\leq
  \sum_{\substack{
      \uu'\in\U^n\setminus\{\uu\}:
      \\
      p_{\sfA,\uu,\uu'}>\frac{\alphaA}{|\im\A|}
  }}
  p_{\sfA,\uu,\uu'}
  +
  \sum_{\substack{
      \uu'\in\U^n\setminus\{\uu\}:
      \\
      p_{\sfB,\uu,\uu'}>\frac{\alphaB}{|\im\B|}
  }}
  p_{\sfB,\uu,\uu'}
  \notag
  \\
  &=
  \betaA+\betaB
  \notag
  \\
  &=
  \beta_{\sfA\sfB},
  \label{eq:proof-AB}
\end{align}
where the first inequality comes from the fact that
$\im\A\times\B\subset\im\A\times\im\B$ and $\sfA$, $\sfB$ are mutually
independent,
and the last inequality comes from the fact that
$p_{\sfA,\uu,\uu'}\leq 1$, $p_{\sfB,\uu,\uu'}\leq 1$.
Since $(\aalpha_{\sfA\sfB},\bbeta_{\sfA\sfB})$ satisfies
(\ref{eq:alpha}) and (\ref{eq:beta}), we have the fact that
$(\bcA\times\bcB,\bp_{\sfA\sfB})$ has an
$(\aalpha_{\sfA\sfB},\bbeta_{\sfA\sfB})$-hash property.
\hfill\QED

\subsection{Proof of Lemma \ref{lem:CRP}:}
\label{sec:proof-CRP}
We have
\begin{align}
  &\pA\lrsb{\lrb{
      A: \lrB{\G\setminus\{\uu\}}\cap\C_A(A\uu)\neq \emptyset
  }}
  \notag
  \\*
  &\leq
  \sum_{\uu'\in\G\setminus\{\uu\}}
  p_A\lrsb{\lrb{
      A: A\uu = A\uu'
  }}
  \notag
  \\
  &\leq |\{\uu\}\cap\lrB{\G\setminus\{\uu\}}|
  +\frac{\left|\{\uu\}||\G\setminus\{\uu\}\right|\alphaA}{|\im\A|}
  + \min\{|\{\uu\}|,\left|\G\setminus\{\uu\}\right|\}\betaA
  \notag
  \\
  &\leq
  \frac{|\G|\alphaA}{|\im\A|} + \betaA,
\end{align}
where the second inequality comes from (\ref{eq:whash})
by letting $\T\equiv\{\uu\}$ and $\T'\equiv\G\setminus\{\uu\}$.
\hfill\QED

\subsection{Proof of Lemma~\ref{lem:BCP}}
\label{sec:proof-BCP}

Let $p_{\sfA,\uu,\uu'}$ be defined as
\[
  p_{\sfA,\uu,\uu'}\equiv p_{\sfA}\lrsb{\lrb{A: A\uu=A\uu'}}.
\]
Then we have
\begin{align}
  &
  E_{\sfA\sfcc}\lrB{
    \lrB{\sum_{\uu\in\T}Q(\uu)\chi(\sfA\uu=\sfcc)}^2
  }
  \notag
  \\*
  &=
  E_{\sfA}\lrB{
    \sum_{\uu\in\T}Q(\uu)\sum_{\uu'\in\T}Q(\uu')\chi(\sfA\uu=\sfA\uu')
    E_{\sfcc}\lrB{\chi(\sfA\uu'=\sfcc)}
  }
  \notag
  \\
  &=
  \frac 1{|\im\A|}
  \sum_{\uu\in\T}Q(\uu)\sum_{\uu'\in\T}Q(\uu')p_{\sfA}\lrsb{\lrb{A: A\uu=A\uu'}}
  \notag
  \\
  &=
  \frac 1{|\im\A|}
  \sum_{\uu\in\T}Q(\uu)\left[
    \sum_{\substack{
	\uu'\in\T\setminus\lrb{\uu}
	\\
	p_{\sfA,\uu,\uu'}\leq \alphaA/|\im\A|
    }}
    Q(\uu')
    p_{\sfA,\uu,\uu'}
    +\sum_{\substack{
	\uu'\in\T\setminus\lrb{\uu}
	\\
	p_{\sfA,\uu,\uu'}>\alphaA/|\im\A|
    }}
    Q(\uu')
    p_{\sfA,\uu,\uu'}
    +Q(\uu)
  \right]
  \notag
  \\
  &\leq
  \frac 1{|\im\A|}
  \sum_{\uu\in\T}Q(\uu)\left[
    \sum_{\substack{
	\uu'\in\T\setminus\lrb{\uu}
	\\
	p_{\sfA,\uu,\uu'}\leq \alphaA/|\im\A|
    }}
    \frac{Q(\uu')\alphaA}{|\im\A|}
    +
    \lrB{
      \sum_{\substack{
	  \uu'\in\T\setminus\lrb{\uu}
	  \\
	  p_{\sfA,\uu,\uu'}>\alphaA/|\im\A|
      }}
      p_{\sfA,\uu,\uu'}
      +1
    }
    \max_{\uu\in\T}Q(\uu)
  \right]
  \notag
  \\
  &\leq
  \frac{Q(\T)^2\alphaA}{|\im\A|^2}
  +\frac {Q(\T)[\betaA+1]\max_{\uu\in\T} Q(\uu)}{|\im\A|},
  \label{eq:EsumPP}
\end{align}
where $\chi(\cdot)$ is defined by (\ref{eq:chi}),
the second equality comes from the fact that
the uniqueness of the value $A\uu'$ implies
\begin{align}
  E_{\sfcc}\lrB{\chi(A\uu'=\sfcc)}
  &=\frac1{|\im\A|}\sum_{\cc}\chi(A\uu'=\cc)
  \notag
  \\
  &=\frac1{|\im\A|}
\end{align}
for any $A\in\A$ and $\uu'\in\U^n$
when the distribution of $\sfcc$ is uniform on $\im\A$.
Then the lemma is shown as
\begin{align}
  E_{\sfA}\lrB{
    \sum_{\cc}
    \left|
      \frac{Q\lrsb{\T\cap\C_{\sfA}(\cc)}}{Q(\T)}
      -\frac 1{|\im\A|}
    \right|
  }
  &=
  E_{\sfA}\lrB{
    \sum_{\cc}\frac 1{|\im\A|}
    \left|
      \frac{Q\lrsb{\T\cap\C_{\sfA}(\cc)}|\im\A|}{Q(\T)}
      -1
    \right|
  }
  \notag
  \\
  &=
  E_{\sfA\sfcc}\lrB{
    \sqrt{
      \lrB{\frac{Q\lrsb{\T\cap\C_{\sfA}(\sfcc)}|\im\A|}{Q(\T)}
	-1}^2
    }
  }
  \notag
  \\
  &\leq
  \sqrt{
    E_{\sfA\sfcc}\lrB{
      \lrB{\frac{Q\lrsb{\T\cap\C_{\sfA}(\sfcc)}|\im\A|}{Q(\T)}
	-1}^2
    }
  }
  \notag
  \\
  &=
  \sqrt{
    \frac{|\im\A|^2}{Q(\T)^2}
    E_{\sfA\sfcc}\lrB{
      \lrB{\sum_{\uu\in\T}Q(\uu)\chi(A\uu=\cc)}^2
    }
    -1
  }
  \notag
  \\
  &\leq
  \sqrt{
    \alphaA-1
    +\frac {[\betaA+1]|\im\A|\max_{\uu\in\T}Q(\uu)}{Q(\T)},
  }
\end{align}
where the third equality comes from the fact that
$\{\C_A(\cc)\}_{\cc\in\im\A}$ is a partition of $\U^n$
and the last inequality comes from (\ref{eq:EsumPP}).

\subsection{Proof of Lemma \ref{thm:hash-linear}}
\label{sec:proof-linear}

For a type $\bt$, let $\C_{\bt}$ be defined as
\[
  \C_{\bt} \equiv \lrb{\uu\in\U^n :\ \bt(\uu)=\bt}.
\]
We assume that $\pA\lrsb{\lrb{A: A\uu=\zero}}$ depends on $\uu$ only
through the type $\bt(\uu)$.
For a given $\uu\in\C_{\bt}$, we define 
\[
  p_{\sfA,\bt}
  \equiv \pA\lrsb{\lrb{A: A\uu=\zero}}.
\]
We use the following lemma, which is proved for the completeness of the
paper.
\begin{lem}[{\cite[Lemma 9]{HASH}}]
Let $(\alphaA,\betaA)$ be defined by (\ref{eq:alpha-linear})
and (\ref{eq:beta-linear}). Then
\begin{align}
  \alphaA
  &=
  |\im\A|\max_{\bt\in \hcH_{\sfA}}p_{\sfA,\bt} 
  \label{eq:alpha2}
  \\
  \betaA
  &=
  \sum_{\bt\in\cH\setminus\hcH_{\sfA}}|\C_{\bt}|p_{\sfA,\bt},
  \label{eq:beta2}
\end{align}
where $\cH$ is a set of all types of length $n$ except for the type
of the zero vector.
\end{lem}
\begin{proof}
We have
\begin{align}
  S(\pA,\bt)
  &=
  \sum_{A}p_A(A)
  \sum_{\substack{
      \uu\in\C_{\bt}:
      \\
      A\uu=\zero
  }}
  1
  \notag
  \\*
  &=
  \sum_{\uu\in\C_{\bt}}
  \sum_{A:A\uu=\zero}
  \pA(A)
  \notag
  \\
  &=
  |\C_{\bt}|p_{\sfA,\bt}.
  \label{eq:proof-linear-lemma1}
\end{align}
Similarly, we have
\begin{align}
  S(p_{\osfA},\bt)
  &=
  |\C_{\bt}|p_{\osfA,\bt}
  \notag
  \\
  &=|\C_{\bt}||\U|^{-l},
  \label{eq:proof-linear-lemma2}
\end{align}
where the last equality comes from the fact that
\begin{align}
  p_{\osfA,\bt}
  &= 
  \frac {|\U|^{[n-1]l}}{|\U|^{nl}}
  \notag
  \\
  &= |\U|^{-l}
\end{align}
because we can find $|\U|^{[n-1]l}$ matrices $\overline{A}$
to satisfy $\overline{A}\uu=\zero$ for a given $\uu\in\C_{\bt}$.
The lemma can be shown immediately from (\ref{eq:alpha-linear}),
(\ref{eq:beta-linear}), (\ref{eq:proof-linear-lemma1}), and
(\ref{eq:proof-linear-lemma2}).
\end{proof}

Now we prove Lemma~\ref{thm:hash-linear}.
It is enough to show (\ref{eq:hash}) because (\ref{eq:alpha}),
(\ref{eq:beta}) are satisfied from the assumption of the lemma.
Since function $A$ is linear, we have
\begin{align}
  p_{\sfA}(\{A: A\uu = A\uu'\})
  &=
  p_{\sfA}(\{A: A[\uu-\uu']=\zero\})
  \notag
  \\
  &=
  p_{\sfA,\bt(\uu-\uu')}
\end{align}
Then, for $\uu\neq\uu'$ satisfying $\bt(\uu-\uu')\in\hcH_{\sfA}$, we have
\begin{align}
  p_{\sfA}(\{A: A\uu = A\uu'\})
  &=
  p_{\sfA,\bt(\uu-\uu')}
  \notag
  \\
  &\leq
  \max_{\bt\in \hcH_{\sfA}}p_{\sfA,\bt} 
  \notag
  \\
  &=
  \frac{\alphaA}{|\im\A|},
\end{align}
where the last inequality comes from (\ref{eq:alpha2}).
Then we have the fact that
$p_{\sfA}(\{A: A\uu = A\uu'\})>\alphaA/|\im\A|$ implies
$\bt(\uu-\uu')\notin\hcH_{\sfA}$.
Finally, we have
\begin{align}
  \sum_{\substack{
      \uu'\in\U^n\setminus\{\uu\}:
      \\
      p_{\sfA}(\{A: A\uu = A\uu'\})>\frac{\alphaA}{|\im\A|}
  }}
  p_{\sfA}\lrsb{\lrb{A: A\uu = A\uu'}}
  &\leq
  \sum_{\substack{
      \uu'\in\U^n\setminus\{\uu\}:
      \\
      \bt(\uu-\uu')\in\cH\setminus\hcH_{\sfA}
  }}
  p_{\sfA,\bt(\uu-\uu')}
  \notag
  \\
  &\leq
  \sum_{\bt\in\cH\setminus\hcH_{\sfA}}
  \sum_{\substack{
      \uu'\in\U^n\setminus\{\uu\}:
      \\
      \bt(\uu-\uu')=\bt
  }}
  p_{\sfA,\bt}
  \notag
  \\
  &\leq
  \sum_{\bt\in\cH\setminus\hcH_{\sfA}}
  |\C_{\bt}|p_{\sfA,\bt}
  \notag
  \\
  &=\betaA,
\end{align}
where the equality comes from (\ref{eq:beta2}).
\hfill\QED

\section*{Acknowledgements}
The author thanks Dr.~S.~Miyake, Prof.~K.~Iwata,
Prof.~T.~Ogawa, and Porf.~H.~Koga for helpful discussions.

\end{document}